\documentclass{article}
\usepackage{graphicx} 
\usepackage{amsfonts, amsthm, amsmath, amssymb, xcolor, nicematrix}
\usepackage{comment}
\usepackage{algorithm}
\usepackage{algpseudocode}
\usepackage{hyperref}
\usepackage{tikz}
\usetikzlibrary{decorations.pathreplacing}
\usepackage[
    commandnameprefix=ifneeded,
    todonotes={colorinlistoftodos,
                prependcaption,
                textsize=small,
                backgroundcolor=orange!10,
                textcolor=black,
                linecolor=orange,
                bordercolor=orange}
]{changes}

\definechangesauthor[name={GA}, color=teal]{GA}
\definechangesauthor[name={R2}, color=red]{R2}
\definechangesauthor[name={GDC}, color=magenta]{GDC}
\definechangesauthor[name={PB}, color=blue]{PB}
\newcommand{\addedg}[1]{\added[id=GA]{#1}}

\usepackage[braket]{qcircuit}
\usepackage{environ}
\newcommand{\myqctmp}[2][0.25]{\Qcircuit @C=#2em @R=#1em @!R}
\NewEnviron{myqcircuit}[1][0.25]{\vcenter{\myqctmp[#1]{0.5} {\BODY}}}
\usepackage{subcaption}
\newtheorem{theorem}{Theorem}
\newtheorem{lemma}[theorem]{Lemma}

\newtheorem{remark}[theorem]{Remark}
\newtheorem{proposition}[theorem]{Proposition}
\newtheorem{definition}[theorem]{Definition}

\providecommand{\keywords}[1]
{
  \small	
  \textbf{\textit{Keywords---}} #1
}
\title{Quantum block encoding for one-pair semiseparable matrices}
\author{G. Antonioli\thanks{Dipartimento di Informatica, Universit\`a di Pisa, Largo Bruno Pontecorvo 3, 56127 Pisa, Italia. Member of the INdAM Research Group GNCS. Email: {\tt giacomo.antonioli@phd.unipi.it} Orcid: {\tt 0009-0000-6687-0357} },
P. Boito\thanks{Dipartimento di Matematica, Universit\`a di Pisa, Largo Bruno Pontecorvo 5, 56127 Pisa, Italia. Member of the INdAM Research Group GNCS. Email: {\tt paola.boito@unipi.it} Orcid: {\tt 0000-0002-3559-393X}},
G. M. Del Corso\thanks{Dipartimento di Informatica, Universit\`a di Pisa, Largo Bruno Pontecorvo 3, 56127 Pisa, Italia. Member of the INdAM Research Group GNCS. 
Email: {\tt gianna.delcorso@unipi.it} Orcid: {\tt 0000-0002-5651-9368}},
M. Porcelli\thanks{Dipartimento di Ingegneria Industriale, Universit\`a degli Studi di Firenze, Viale Morgagni 40/44, 50134 Firenze, Italia. Member of the INdAM Research Group GNCS. Email: {\tt margherita.porcelli@unifi.it} Orcid: {\tt 0000-0003-0183-1204}}\ \thanks{ISTI--CNR, Via Moruzzi 1, Pisa, Italia} }
\date{June 2026}

\begin{document}

\maketitle

\begin{abstract}
Quantum block encoding (QBE) is a crucial step in the development of most quantum algorithms, as it provides an embedding of a given matrix into a suitable larger unitary matrix. Historically, the development of efficient techniques for QBE has mostly focused on sparse matrices; less effort has been devoted to data-sparse (e.g., rank-structured) matrices. 

In this work we examine a particular case of rank structure, namely, one-pair semiseparable matrices. We present a new block encoding approach that relies on a suitable factorization of the given matrix as the product of triangular and diagonal factors. To encode the matrix, the algorithm needs $2\log(N)+7$ ancillary qubits. 

Assuming that the data input oracles can be implemented with polylogarithmic depth, or that a QRAM input model is available, our proposed method requires $\mathcal{O}({\rm polylog} (N))$ time and has an error of $\mathcal{O}(N^2)$, where $N$ is the matrix size.

\end{abstract}
\section{Introduction}

Block encoding plays a fundamental role in the design of quantum algorithms, particularly when quantum numerical linear algebra is involved. Since operators acting on a quantum system -- such as gates in a quantum circuit -- are necessarily unitary, quantum algorithms cannot work directly on arbitrary matrices, but are instead restricted to unitary matrices. Roughly speaking, block encoding consists in a suitable embedding of a given matrix $A$ in a larger unitary matrix $U$ that can be expressed as the product of a small number of simple unitary transformations, and such that $A$ is found in the top left block of $U$. One can then perform computations involving $A$ via a quantum circuit that encodes $U$.

In many applications, $A$ is sparse or otherwise structured. Exploiting this structure is crucial in the development of effective block encoding strategies. For instance, the well-known HHL algorithm \cite{harrow2009quantum} for the solution of a linear system $Ax=b$ reaches an exponential speedup with respect to classical conjugate gradient only if $A$ is sparse and efficiently encoded; otherwise the speedup becomes quadratic \cite{wossnig2018quantum}.

In this context it is interesting to explore block encoding techniques for (possibly dense) matrices endowed with some kind of rank structure. Rank structures often stem from applications and reflect fundamental properties of the problem under study, such as localization or shift invariance. They usually fall into one of two main classes: displacement structure (i.e., the output of a certain linear operator applied to $A$ is a low-rank matrix) and off-diagonal rank structure (i.e., off-diagonal blocks of $A$ have low rank; this condition is at the basis of several classes of structured matrices such as semiseparable, quasiseparable, or hierarchical semiseparable). 

Block encoding techniques for displacement structured matrices are proposed in \cite{wan2021block}. In this work the authors address at first the case of Toeplitz matrices and then apply the LCU (linear combination of unitaries) method to block encode a subclass of Toeplitz-like matrices; similar considerations apply of course to Hankel and Hankel-like matrices.  Sparsity structure is exploited for efficient block encoding e.g. in \cite{camps2024explicit}.

To the best of our knowledge, no effort has been made so far towards the block encoding of semiseparable and related matrix structures. In this work, we focus on one-pair semiseparable matrices and develop an efficient block encoding algorithm. The method relies on a suitable factorization of the matrix into components that are easily block encoded, and by combining these, we can efficiently obtain the desired structured matrix in $\mathcal{O}({\rm polylog} (N))$ time, plus two calls to an input oracle. If the oracle can be implemented with polylogarithmic depth, or a QRAM input model is available, our proposed method requires $\mathcal{O}({\rm polylog} (N))$ time and has an error of $\mathcal{O}(N^2)$, where $N$ is the matrix size.

For comparison purposes, we consider the FABLE package \cite{fable} for general matrix block encoding. Note however that in FABLE the gate complexity of matrix block encoding scales up to $\mathcal{O}(N^2)$, where $N$ is the matrix size, reflecting the cost of encoding all independent elements.


While FABLE's compression can exploit implicit structure, its maximum $\mathcal{O}(N^2)$ complexity is not optimal for rank-structured matrices. An $N\times N$ one-pair matrix, for instance, is completely defined by its $2N$ generators, and basic operations such as matrix multiplication or inversion have an $\mathcal{O}(N)$ complexity in a classical setting. 
Even in a worst-case scenario where our input oracle is implemented as in FABLE, the required loading time scales like $\mathcal{O}(N)$ and the overall cost remains linear in $N$, instead of quadratic. We validate our approach by comparing it to FABLE, proving that our factorization technique successfully delivers a block encoding for one-pair semiseparable matrices with demonstrably reduced complexity and error compared to the general framework.



Sections ~\ref{sec:Onepairdef} and ~\ref{sec:qubit} are devoted to introductory material. Definitions and notation for one-pair semiseparable matrices are given in Section ~\ref{sec:Onepairdef}. We recall in~\ref{sec:qubit} the main quantum computing concepts used in the paper, with an
emphasis on the aspects relevant to numerical linear algebra; we follow the
standard notation of~\cite{lin2022lecture}. Section ~\ref{sec:QBE} briefly reviews the main features of unstructured block encoding in the FABLE package, and then focuses on block encoding of diagonal matrices. The overall algorithm for one-pair matrices and its building blocks are analyzed in Section~\ref{sec:building_blocks}, whereas Section ~\ref{Sec:NE} presents numerical experiments. 

The block encoding of diagonal matrices (Section~\ref{sec:diag}) adapts results from~\cite{li2023efficient,lin2022lecture}; the encoding proposed in Section~\ref{sec:QBEL} builds on~\cite{camps2024explicit}, where the circuit for the shift operator is constructed. The main original contributions of this work are the 
overall encoding strategy for one-pair semiseparable matrices and the complexity results presented in Theorem~\ref{lemma:Dz}, Theorem~\ref{teo:FE}.

\section{One-pair semiseparable matrices}
\label{sec:Onepairdef}
The main goal of the present work is the design and implementation of a block encoding strategy for one-pair semiseparable matrices. One-pair matrices were originally introduced in \cite{gantmacher2002oscillation} in the context of operator theory; other names used in the literature include single-pair and generator-representable semiseparable matrices \cite{vandebril2005note}. 

One-pair matrix structure is a particular case of off-diagonal low-rank structure. Variants of this class of matrix structure -- such as semiseparable, quasiseparable, HODLR, hierarchical -- arise often in applications and have been widely exploited to design fast algorithms for solving linear systems, computing eigenvalues and performing many other crucial tasks in numerical linear algebra \cite{vandebril2008matrix, eidelman2014separable, bebendorf2008hierarchical, ChanGu06, massei2020hm}.  

\begin{definition}
    Given vectors $u=[u_i]_{i=0,\ldots,N-1}$, $v=[v_j]_{j=0,\ldots,N-1}$ of dimension $N$, the one-pair matrix associated with $u, v$ is a symmetric $N\times N$ matrix $S = S(u, v)$ such that
    $$
    S_{i,j}=\left\{\begin{array}{cc}
    u_i v_j     & {\rm if } \, i\leq j  \\
    u_j v_i     & {\rm if } \, i\geq j
    \end{array}\right. .
    $$  
    
\end{definition}
 
In other words, we have
$$
S = S(u, v) = \left[ \begin{array}{ccccc}
u_0 v_0 & u_0 v_1 & u_0 v_2 & \cdots & u_0 v_{N-1} \\
u_0 v_1 & u_1 v_1 & u_1 v_2 & \cdots & u_1 v_{N-1} \\
u_0 v_2 & u_1 v_2 & u_2 v_2 & \cdots & u_2 v_{N-1} \\
\vdots  & \vdots  & \vdots  & \ddots & \vdots \\
u_0 v_{N-1} & u_1 v_{N-1} & u_2 v_{N-1} & \cdots & u_{N-1} v_{N-1}
\end{array} \right].
$$
A well-known property of one-pair matrices is their connection to Jacobi matrices, as stated in the following result \cite{gantmacher2002oscillation, vandebril2005note}.
\begin{theorem}
    Let $A$ be an $N\times N$ symmetric tridiagonal matrix of the form
    $$
    A=\begin{bmatrix}
        a_0 & b_0 \\ b_0 & a_1 & b_1\\ & \ddots & \ddots & \ddots \\
        && b_{N-3} & a_{N-2}& b_{N-2}\\ &&& b_{N-2} & a_{N-1}
    \end{bmatrix}
    $$
    with $b_i\neq 0$ for $i=0,\ldots,N-2$. Then the inverse of $A$ is a one-pair matrix.

    Conversely, the inverse of an invertible one-pair matrix is an irreducible symmetric tridiagonal matrix.
\end{theorem}

In order to implement the block encoding of one-pair semiseparable matrices, we use the following factorization (see, e.g., the survey \cite{fasino}, or the proof of Theorem 1 in \cite{fasino2002structural}, and references therein).  
\begin{proposition} \label{prop:decomoposition}
    Given vectors $u, v$ with $u_i\neq 0$ for $i=0,\ldots, N-1$, define $D_u={\rm diag}(u_0, u_1, \ldots, u_{N-1})$ and let $z_i=v_i/u_i$, $i=0\ldots,{N-1}$. Then it holds
\begin{equation}\label{eq:mainfact}
S(u, v)=D_u L \Delta_z L^T D_u,
\end{equation}
where $\Delta_z={\rm diag}(z_0, z_1-z_0, \ldots, z_{N-1}-z_{N-2})$ and
$L= \begin{bmatrix}
1 &   &  \\
\vdots & \ddots &  \\
1 & \cdots & 1
\end{bmatrix}$.
\end{proposition}
For a sketch of proof, observe that we have
$$
S(e, z)= D_u^{-1}S(u, v) D_u^{-1},
$$
where $e= (1, \dots,1)^T$. Moreover, it is easily seen that $S(e,z)=L\Delta_zL^T$, from which the thesis  follows immediately. Also note that, if $u_i=0$ for some index $i$, then $S(u,v)$ does not admit a triangular factorization.

We will see that, under the assumption that the considered one-pair matrices are real valued and have dimensions that are powers of two,  all the factors in \eqref{eq:mainfact} can be efficiently block encoded assuming access to a quantum-accessible data
structure such as QRAM~\cite{giovannetti2008quantum,kerenidis2017quantum} and efficient oracles for querying the vectors $u$ and $v$.

\section{Quantum circuits and block encodings}
\label{sec:qubit}
We use the {\em bra-ket} Dirac $\bra{\cdot}$ and $\ket{\cdot}$ notation to denote row
and column vectors, respectively.  The two canonical basis vectors of
$\mathbb{C}^2$ are written $\ket{0} = [1,0]^T$ and
$\ket{1} = [0,1]^T$.
A \emph{qubit} is a unit vector in $\mathbb{C}^2$, i.e., the two-dimensional complex Hilbert space with  basis $\{ \ket{0}, \ket{1} \}$; its most general state is a
superposition $\alpha\ket{0}+\beta\ket{1}$ with $|\alpha|^2+|\beta|^2=1$.

A joint system of qubits is described by the tensor product of the individual spaces, so  $\ket{x}\ket{y}=\ket{xy}= \ket{x} \otimes \ket{y} \in \mathbb{C}^2 \otimes\mathbb{C}^2 \cong \mathbb{C}^4.$ 
More generally, a quantum register of $n$ qubits lives in $(\mathbb{C}^2)^{\otimes n}\cong\mathbb{C}^N$,
$N=2^n$, so $n$ physical qubits describe a vector in an exponentially large
space. The standard basis of this space, called {\em computational basis} is $\{\ket{j}\}_{j=0}^{N-1}$, where
$j$ is the integer whose $n$-bit binary representation labels the basis state, and $\ket{j}$ is the $j$-th column of the $N\times N$ identity matrix $I_N$.
A general state of on $n$ qubit register,
\[
  \ket{\psi} = \sum_{j=0}^{N-1}\alpha_j\ket{j},\qquad \sum_j|\alpha_j|^2=1,
\]
encodes $2^n$ complex amplitudes simultaneously. In general this state is
\emph{entangled} and cannot be written as a tensor product of individual
qubit states, which is one of the key resources exploited by quantum
algorithms.
\paragraph{Gates.}
Quantum computation proceeds by applying unitary matrices, called gates, to
qubit registers. In this paper we use the following single-qubit gates whose matrix forms are
\[
  H = \frac{1}{\sqrt{2}}\begin{pmatrix}1&1\\1&-1\end{pmatrix},\quad
  X = \begin{pmatrix}0&1\\1&0\end{pmatrix},\quad
  Z = \begin{pmatrix}1&0\\0&-1\end{pmatrix},\quad
  R_y(\theta) = \begin{pmatrix}\cos\frac{\theta}{2}&-\sin\frac{\theta}{2}\\
                               \sin\frac{\theta}{2}&\cos\frac{\theta}{2}\end{pmatrix}.
\]

As  2-qubit gates  in this paper we use the CNOT gate and the SWAP gate : 
$$\operatorname{CNOT}=\begin{bmatrix}
    1 & 0&0&0\\0&1&0&0\\0&0&0&1\\0&0&1&0
\end{bmatrix}, \quad
{\rm{SWAP}}=\begin{bmatrix}
    1 & 0&0&0\\0&0&1&0\\0&1&0&0\\0&0&0&1
\end{bmatrix},$$
The CNOT gate inverts the state of the second qubit when the state of the first one is $\ket{1}$.
The SWAP gate exchanges the quantum states of two adjacent qubits. It is possible to introduce another operator which allows to swap any two non-necessarily adjacents qubits (we need to use three CNOT gates). 

In general, a controlled-$U$ gate applies a unitary $U$ to a target register
conditionally on a control qubit being in state $\ket{1}$. We also use multi-controlled gates, where the gate is applied to target qubits only if the control qubits are simultaneously in state $\ket{1}$. A
\emph{quantum circuit} is a sequence of such gates.

\paragraph{Measurement.}

In quantum mechanics, extracting classical information from a quantum state is done via a measurement.
A measurement is described by a set of projectors $\{P_m\}$ satisfying $\sum_m P_m = I$ and $P_m P_{m'} = \delta_{mm'} P_m$.
When a state $\ket{\psi}$ is measured, outcome $m$ occurs with probability $\|P_m\ket{\psi}\|^2 = \bra{\psi} P_m \ket{\psi}$, and the state collapses to $P_m\ket{\psi} / \|P_m\ket{\psi}\|$.
The most common case is when the measurement is performed in the computational basis; in this case $P_j = \ket{j}\bra{j}$ for basis states $\ket{j}$. In particular, a single qubit in state $\alpha\ket{0}+\beta\ket{1}$ yields
outcome $0$ with probability $|\alpha|^2$ and outcome $1$ with probability
$|\beta|^2$.
Measurement is how a quantum computer outputs a result; repeated measurements allow estimating the amplitudes with frequencies.
This paper focuses on unitary circuit constructions, but the reader may assume that any needed information can be obtained via suitable measurements with an overhead depending on the required accuracy.


\paragraph{Cost models}
\label{par:costmodel}
Quantum algorithms are typically analyzed within two complexity models. In the {\em query complexity} model~\cite{lin2022lecture,KLM}, the cost of an algorithm is measured by the number of calls to a black-box oracle $\mathcal{O}$
encoding the input; for instance, a unitary that maps $\mathcal{O}:\ket{i}\ket{0}\to \ket{i}\ket{x_i}$ for some data $x_i$. In general, a computation in this model is one which computes some function $f(x_0, \ldots, x_{N-1})$ accessing the oracle $\mathcal{O}$, and the query complexity of $f$ is the minimal number of queries to the oracle to output the correct value $f(x)$ for every $x$. 
 In the {\em circuit complexity model}~\cite{nielsen2010quantum,kitaev2002}, the cost is measured by three quantities: the size of the circuit, i.e.\ the total number of elementary gates drawn from a fixed universal gate set,  the depth, i.e.\ the length of the longest path from input to output in the directed acyclic graph of gate dependencies, and the width, which corresponds to the total number of qubits used in the circuit. Depth determines the runtime under the assumption of parallel execution of independent gates. We say that a circuit has polylogarithmic depth if its
gate depth is $\mathcal{O}(\rm{polylog}(N))$; this is the efficiency target
throughout this paper, as it corresponds to an exponential speedup over
classical sequential computation of typical linear algebra problems which usually require at least a linear number of operations in $N$. 

A query complexity result translates into a
circuit complexity result once the oracle is decomposed into elementary gates;
this cost may depend on how classical data is
loaded into the quantum state, for instance via QRAM. This cost can dominate the overall gate
count and depth and must be accounted for explicitly.

Throughout this work we will state costs in both models where relevant, distinguishing oracle queries from gate complexity to make the assumptions behind each result explicit. 
Whenever we claim polylogarithmic complexity, the claim refers to
gate complexity under the assumption that oracles can be implemented by circuit of depth $\mathcal{O}({\rm polylog}(N))$.

  
\subsection{Block encoding}\label{subsec:be}
A block encoding of a matrix $A$ is a unitary $U_A$ that encodes
$A/\alpha$ as its top left block, where $\alpha\ge \|A\|_2$ is a scaling factor.  We give below a formal definition.

\begin{definition}(Block encoding). Given an $n$-qubit matrix $A$ ($N = 2^n$), if we can find $\alpha, \varepsilon \in \mathbb{R}_+$ and an $(m+n)$-qubit
unitary matrix $U_A$ so that

$$
\| A - \alpha (\langle 0^m | \otimes I_N) U_A (|0^m \rangle \otimes I_N)\|_2 \le \varepsilon
$$
then $U_A$ is called an $(\alpha, m, \varepsilon)$-block-encoding of $A$. In particular, when the block encoding is exact with  $\varepsilon = 0$, $U_A$ is
called an $(\alpha, m)$-block-encoding of $A$.
\end{definition}

Here $m$ is the number of ancilla qubits used to block encode $A$, and the expression $(\langle 0^m | \otimes I_N) U_A (|0^m \rangle \otimes I_N)$ should
be interpreted as taking the upper-left $2^n \times 2^n$ matrix block of $U_A$.

Matrix block encoding allows us to perform several matrix operations, such as matrix-vector 
and matrix-matrix products and matrix inversion (examples can be found, e.g. in  \cite{camps2024explicit}). 

When matrices are small, e.g. $N=2$, explicit forms for $U_A$ can be given \cite{camps2024explicit}, but these are hardly extensible to a large setting.

As an example, a strategy to construct a block encoding of a (scaled) $n$-qubit Hermitian matrix $A$ is defining $U_A$ as follows
\begin{equation}\label{eq:BE}
    U_A =\begin{bmatrix}
A & (I - A^\dagger A)^{1/2} \\
(I - A^\dagger A)^{1/2} & - A
\end{bmatrix} \mbox{ or }
\begin{bmatrix}
A & -(I - A^\dagger A)^{1/2} \\
(I - A^\dagger A)^{1/2} &  A
\end{bmatrix}.
\end{equation}

This formula can be used for a 2×2 matrix $A$, but it is unpractical for larger matrices. Indeed, this definition involves the computation of a matrix square root and the diagonalization of $A$, which are computationally expensive. Moreover, the resulting unitary matrix $U_A$ must be decomposed into a sequence of elementary quantum gates. This decomposition would likely require a non-polylogarithmic gate depth, making it infeasible for large matrices. As such, this approach is unimplementable on a quantum computer using an efficient number of quantum gates.

It is important to keep the scaling factor $\alpha$ as close to $\|A\|_2$ as possible, as this directly impacts the success probability when we perform a measurement to extract $\frac{A}{\alpha} \ket{b}$. Specifically, when computing the matrix-vector product $Ab$, we must run the circuit that implements $U_A$ with the initial state $\ket{0}^m \ket{b}$ and discard all measurements that do not yield $0$ in the first $m$ qubits. The probability of measuring $0$ in the first $m$ qubits is given by $\frac{\|A\ket{b}\|^2}{\alpha^2}$; see, e.g., \cite{lin2022lecture} for a brief discussion. Therefore, if $\alpha$ is too large compared with $\|A\|_2$, the probability of measuring $0$ becomes very low, requiring many
repetitions or Amplitude Amplification~\cite{brassard2000quantum} at additional cost. It is therefore
important to keep $\alpha$ as close to $\|A\|_2$ as possible.

Achieving a  block encoding realizable by a circuit of
$\mathcal{O}(\rm{polylog}(N))$ depth and $\mathcal{O}(\rm{polylog}(N))$
ancilla qubits, is in general a nontrivial task, and is the central
algorithmic challenge addressed in this paper. 

\subsection{Block encoding of a linear combination of matrices}

The Linear Combination of Unitaries (LCU) \cite{Ch12} method is a powerful technique for block encoding non-unitary matrices. It works by decomposing a matrix $A$ as a weighted sum of unitary matrices, represented by the formula $A = \sum_{i=1}^m \alpha_i U_i$. The $\alpha_i$ are a set of positive coefficients, and the $U_i$ are easily implementable unitary matrices. Any negative sign of complex phase factors can be absorbed into the definition of the corresponding $U_i$ terms, so we can assume the coefficients are positive without loss of generality.

To implement this on a quantum computer, two key components are constructed. First, a select oracle $U = \sum_{i=0}^m \ket{i}\bra{i} \otimes U_i$ is created using $a=\log_2(m)$ control qubits. This oracle applies the unitary $U_i$ to a target register only when the control qubits are in the state $\ket{i}$. Second, a unitary $V$ is constructed to encode the coefficients of the linear combination, such that
$$
V \ket{0}^{\otimes a} = \frac{1}{\sqrt{\| \alpha\|_1 }} \sum_{i=1}^m \sqrt{\alpha_i} \ket{i}.
$$
A block encoding for $A$ can then be realized using the unitary $W = (V^\dagger \otimes I_n) U (V \otimes I_n)$, which effectively prepares the coefficient state, applies the correct unitary, and then un-prepares the coefficient state.

In our application, we only need to combine two matrices at a time. We assume
that the two matrices have the same dimension, and that this dimension is a
power of two. Thus, let \(A\) and \(B\) be two \(N\times N\) matrices, with
\(N=2^n\). Suppose that \(U_A\) is an \((\alpha,m,\epsilon_A)\)-block encoding
of \(A\), and that \(U_B\) is a \((\beta,m,\epsilon_B)\)-block encoding of
\(B\).

Using one additional control qubit, the circuit prepares a uniform superposition of two branches: one applying $U_A$ and the other applying $U_B$. After a final Hadamard on the control qubit, projecting the control and the block-encoding ancillas onto the all-zero state gives the normalized block ${1}/{2}\left({A}/{\alpha}+{B}/{\beta}\right)$. Therefore, the first circuit in Figure~\ref{fig:LCU_sum} gives a $(2\alpha\beta,m+1,2(\alpha\epsilon_B+\beta\epsilon_A))$-block encoding of the sum of the two original block encodings.

The second circuit, shown in Figure~\ref{fig:LCU_diff}, is obtained by selecting the opposite output branch after the final Hadamard through the additional $X$ gate. This gives the normalized block ${1}/{2}\left({A}/{\alpha}-{B}/{\beta}\right)$ and therefore a $(2\alpha\beta,m+1,2(\alpha\epsilon_B+\beta\epsilon_A))$-block encoding of the difference of the two original block encodings.

\begin{figure}[h!]
    \begin{subfigure}[t]{0.45\textwidth}
        \centering
        $
\begin{myqcircuit}
\lstick{\ket{0}}&\qw & \gate{H} &\ctrlo{1} & \qw & \ctrl{1} &\qw& \gate{H} & \qw \\
\lstick{\ket{0}}& {/}^m \qw & \qw &\multigate{1}{U_A}&\qw&\multigate{1}{U_B} & \qw & \qw & \qw\\
\lstick{\ket{i}}& {/}^n\qw & \qw & \ghost{U_A} &\qw & \ghost{U_B} & \qw& \qw & \qw
\end{myqcircuit}
        $
        \caption{Quantum circuit for summing block encoded matrices.}
        \label{fig:LCU_sum}
    \end{subfigure}
    \hfill
    \begin{subfigure}[t]{0.45\textwidth}
        \centering
        $
\begin{myqcircuit}
\lstick{\ket{0}}&\qw & \gate{H} &\ctrlo{1} & \qw & \ctrl{1} &\qw &\gate{H} & \qw& \gate{X} & \qw \\
\lstick{\ket{0}}& {/}^m \qw & \qw &\multigate{1}{U_A}&\qw&\multigate{1}{U_B} & \qw & \qw & \qw& \qw& \qw \\
\lstick{\ket{i} }&{/}^n\qw & \qw & \ghost{U_A} &\qw & \ghost{U_B} & \qw& \qw & \qw& \qw & \qw
\end{myqcircuit}
        $
        \caption{Quantum circuit for subtracting block encoded matrices.}
        \label{fig:LCU_diff}
    \end{subfigure}
    \caption{Quantum circuits illustrating how to combine two block encodings for matrices $A$ and $B$ to create a block encoding for their sum and difference, respectively.}
    \label{fig:LCU}
\end{figure}
In Section~\ref{sec:building_blocks} we use the concepts introduced here  to implement the encoding of $\Delta_z$ introduced in Proposition~\ref{prop:decomoposition}.

\subsection{Block encoding of a matrix product}
The block encoding of the product of two matrices $A$ and $B$ is performed essentially by multiplying the block encodings of $A$ and $B$ with some additional qubits, as explained in the following result \cite{gilyen2019quantum}.

\begin{lemma} (Product of block-encoded matrices)\label{lemma:prod} 
If $U_A$ is an $(\alpha, a, \delta)$-block-encoding of an $n$-qubit
operator $A$, and $U_B$ is a $(\beta, b, \varepsilon)$-block-encoding of an $n$-qubit operator $B$, then $(I_b\otimes U_a)(I_a\otimes U_B)$ is
an $(\alpha\beta, a+b, \alpha\varepsilon+\beta\delta)$-block-encoding of $AB$.
\end{lemma} 

The Kronecker product used in the lemma differs from the standard notation in the ordering of the ancilla and signal qubits, since the Identity operators act on each other ancilla qubits. We can express it in the standard form by using a swap operator which exchanges the $i$-th and $j$-th qubits. 
We can now write the circuit of the product as follows: 

\begin{center}
$
\begin{myqcircuit}
\lstick{\ket{0}}&{/}^a\qw &\qw&\multigate{2}{U_A  U_B} & \qw  &&& 
\lstick{\ket{0}}& {/}^a\qw &\qw& \qw& \qw & \qswap & \qw & \qw & \qw	\\
\lstick{\ket{0}}&{/}^b\qw  &\qw&\ghost{U_A  U_B} & \qw  &\push{\rule{.3em}{0em}=\rule{1.5em}{0em}}&& \lstick{\ket{0}}	& {/}^b\qw&\qw&\multigate{1}{U_B}   & \qw	   & \qswap \qwx   & \qw	& \multigate{1}{U_A}   & \qw	\\
\lstick{\ket{i}} &{/}^n\qw &\qw&\ghost{U_A  U_B} & \qw &&& \lstick{\ket{i}}	&{/}^n\qw&\qw&\ghost{U_B}           & \qw	   & \qw           & \qw	& \ghost{U_A} 	       & \qw	\\
\end{myqcircuit}
$
\end{center}
\hfill \break
Lemma~\ref{lemma:prod} allows us to perform a sequence of block-encoded matrix multiplications. It is also important to note that the previously defined SWAP operator has as its range, for the product of the sequence, the number of auxiliary qubits of the second block-encoded matrix. For instance, consider three matrices A, B, and C, along with their corresponding unitaries
encoding these matrices, $U_A, U_B$, and $U_C$. We obtain the block encodings: $(\alpha,a,\varepsilon_A)$, $(\beta,b,\varepsilon_B)$ and $(\gamma,c,\varepsilon_C)$. In a simple case where each encoding requires just one auxiliary qubit, the block encoding of $BC$ is$(\beta\gamma,2,\beta\varepsilon_C+\gamma\varepsilon_B)$. To obtain $ABC$, we need to swap a qubits (in this example $a=1$) and then block encode $A$. This results in the block encoding$(\alpha\beta\gamma,3,\alpha\beta\varepsilon_C+\alpha\gamma\varepsilon_B+ \beta\gamma\varepsilon_A)$ of $ABC$. The circuit to perform this example is as follows.
\newline
\begin{center}
$
\begin{myqcircuit}
\lstick{\ket{0}}	& {/}^a \qw  & \qw	& \qw	                & \qw      & \qw	       & \qw	& \qw	               & \qw	&  \qswap      &   \qw  & \qw                 & \qw  \\
\lstick{\ket{0}}	& {/}^b \qw	 & \qw	& \qw                & \qw      & \qswap		   & \qw	& \qw	               & \qw	&  \qw \qwx    &   \qw  & \qw                 & \qw \\
\lstick{\ket{0}}    & {/}^c \qw	 & \qw	& \multigate{1}{U_C}    & \qw	   & \qswap \qwx   & \qw	& \multigate{1}{U_B}   & \qw	&  \qswap \qwx &   \qw  & \multigate{1}{U_A}  & \qw  \\
\lstick{\ket{i}}	& {/}^n \qw& \qw   & \ghost{U_C}           & \qw	   & \qw           & \qw	& \ghost{U_B} 	       & \qw	&  \qw         &   \qw  & \ghost{U_A}         & \qw \\
\end{myqcircuit}
$
\end{center}
In Section~\ref{sec:building_blocks} we will use the concepts introduced here to implement a chain of products  for the final block-encoding.

\section{Overview on QBE for unstructured and structured matrices}\label{sec:QBE}

In this section, we present a block encoding technique for general unstructured matrices based on an approximate representation of the matrix. This methodology will be used in Section~\ref{Sec:NE} to compare our algorithm in terms of the error. We also discuss block encodings for several matrix structures that serve as building blocks for the block encoding of a one-pair matrix via equation~\eqref{eq:mainfact}.

 \subsection{The unstructured case: FABLE} \label{subsec:fable}

For unstructured matrices, FABLE~\cite{fable} is a state-of-the-art method for block encoding that operates in an explicit-access model: each matrix entry is loaded directly into the circuit via a dedicated rotation gate, so the full cost of data encoding is accounted for in the circuit complexity. Without any compression, this requires one multi-controlled rotation per entry, giving an overall complexity of $\mathcal{O}(N^2)$. FABLE reduces this cost through a classical preprocessing step based on a Walsh–Hadamard transform applied to the angles corresponding to the matrix entries, at an additional classical cost of $\mathcal{O}(N\log N)$. In the Walsh–Hadamard domain the angle vector is typically sparse, and two compression strategies are available: either a cutoff threshold $\delta_c$
is applied, zeroing all transformed angles below $\delta_c$, or a fixed compression rate is enforced by retaining only the largest-magnitude transformed angles. To avoid deep multi-controlled rotations, a Gray code ordering is applied before the circuit is assembled, reducing all rotations to single-controlled gates. The result is an $(N,\log N+1,N^3\delta_c)$-block encoding of the original matrix, where the error term reflects the truncation in the Walsh–Hadamard domain and does not account for errors in the original entries.

\subsection{Block encoding of diagonal matrices and their inverses}\label{sec:diag}

For the block encoding of diagonal matrices and their inverses, we rely on the following result (Lemma 10.10 of~\cite{AB09}), which we report below to make this paper self contained.

\begin{lemma} \label{lem:10:10}
If $f :\{0, 1\}^N \to \{0, 1\}^M$ 
is computable by a Boolean circuit of size $S$ then there is a sequence of $2S+M+N$ quantum elementary gates computing the mapping $\ket{x} \ket{0}^{2M+S} \to \ket{x} \ket{f(x)}\ket{0}^{S+M}$.
\end{lemma}

The following technical lemma, derived from Proposition 4 of~\cite{li2023efficient} and Chapter 4 of~\cite{lin2022lecture}, provides a circuit-based method to encode a real value $\alpha \in [-1,1]$ as the amplitude of $\ket{0}$ of an ancilla qubit. The key idea is to represent $\alpha$ via the angle $\vartheta = \frac{1}{\pi}\arcsin(\alpha)$, so that $\alpha = \sin(\pi\vartheta)$, and then exploit the binary expansion of $\vartheta$ to implement the encoding through a sequence of controlled $R_y$ rotations.
\begin{lemma} \label{lemma:theta}
Let $\vartheta\in\mathbb{R}$ with $|\vartheta|<1$.
 Assume that $\vartheta$ has the following exact $t$-bits binary expansion 
$$
\vartheta=(-1)^{\vartheta^{\,\rm sgn}}(0.\vartheta_{t-1}\vartheta_{t-2}\dots\vartheta_0)=(-1)^{\vartheta^{\,\rm sgn}} \sum^{t-1}_{j=0}\vartheta_j2^{-(t-j)}, \quad \vartheta_j \in \{0, 1\}.
$$

Then, with the convention that $\ket{{\vartheta^{\,\rm sgn}}}=\ket{0}$ if $\vartheta \ge 0$ and $\ket{{\vartheta^{\,\rm sgn}}}=\ket{1}$  if $\vartheta<0$, and  that $\ket{\vartheta}= \ket{\vartheta_{t-1} \vartheta_{t-2}\ldots \vartheta_0}$ there exists 
    a unitary $U_{\vartheta}$ acting on $t+2$ qubits such that 
    $$
    U_\vartheta \ket{0}\ket{\vartheta^{\,\rm sgn}}\ket{\vartheta} =  (-1)^{\vartheta^{\,\rm sgn}}\left(\sin(\pi \vartheta) \ket{0} +\cos(\pi\vartheta)\ket{1}\right)\ket{\vartheta^{\,\rm sgn}}\ket{\vartheta}.$$
\end{lemma}

\begin{figure}[H]
\begin{center}
$
\begin{myqcircuit}
\lstick{\ket{0}}            & \qw  & \gate{R_y(\pi)} & \gate{R_y(\pi/2)} & \qw &\qw &     &     & \ldots &     &     &     & \qw & \gate{R_y(\pi/2^{t-1})} & \qw & \qw \\
\lstick{\ket{\vartheta^{\,\rm sgn}}} & \gate{Z} & \qw               & \qw                 & \qw &\qw & \qw & \qw & \qw    & \qw & \qw & \qw & \qw & \qw                       & \qw        & \qw \\
\lstick{\ket{\vartheta_{t-1}}} & \qw       & \ctrl{-2}         & \qw                 & \qw &\qw & \qw & \qw & \qw    & \qw & \qw & \qw & \qw & \qw                       & \qw        & \qw \\
\lstick{\ket{\vartheta_{t-2}}} & \qw       & \qw               & \ctrl{-3}           & \qw &\qw & \qw & \qw & \qw    & \qw & \qw & \qw & \qw & \qw                       & \qw        & \qw \\
\lstick{\dots}  \\
\lstick{\ket{\vartheta_0}}     & \qw       & \qw               & \qw                 & \qw &\qw & \qw & \qw & \qw    & \qw & \qw & \qw & \qw & \ctrl{-5}                 & \qw         & \qw
\end{myqcircuit}
$
\end{center}
\caption{Controlled rotations of the powers of two}\label{fig:thetaf}
\end{figure}

\begin{proof}
In Figure~\ref{fig:thetaf} is the circuital implementation of $U_\vartheta$. All the rotations are controlled by the binary representation of $\vartheta$ and act on the ancillary qubit.  The  $Z$ gate is employed to accurately encode the sign of $\vartheta$.

\end{proof}

In particular, using Lemma~\ref{lemma:theta} we can block encode a diagonal matrix $D$ as well as its inverse by encoding the entries of $D$ and its inverse as angles.

\begin{theorem}\label{diag_encoding}
Let $D\in \mathbb{R}^{N\times N}$ be a diagonal matrix, 
$D={\rm diag}(d_0, d_1, \ldots, d_{N-1})$ with $N=2^n$. 
Let $c=\frac{\sqrt{\pi^2-1}}{\pi}$, $M=\max_{i=0,\ldots,N-1}|d_{i}|$, 
and let
\[
  t = \left\lceil \log_2\!\left(\frac{M\pi}{c\varepsilon}\right)\right\rceil + 1.
\]
Assume that:
\begin{description}
    \item[{\em (i.)}]  There exists an oracle $O_D$ such that
    \[
      O_D\ket{0}_{t+1}\ket{i} 
      = \ket{d_i^{\,\rm sgn}}\Bigl|\tfrac{c|\tilde{d}_{i}|}{M}\Bigr\rangle_{t}\ket{i},
    \]
    where $\ket{d_i^{\,\rm sgn}}$ encodes the sign of $d_i$ and 
    $|\tilde{d}_i - d_i| \leq 2^{-t}M/c$.

  \item[{\em (ii.)}] There exists a classical Boolean circuit of depth $\mathcal{O}({\rm poly}(t))$ 
    that computes a $t$-digit fixed-point approximation of 
    $x \mapsto \frac{1}{\pi}{\rm arcsin}(x)$ for $x\in[0,c]$ 
    given with $t$-digit precision.

\end{description}

Then an $(M/c,\, 1,\, \varepsilon)$-block encoding of $D$ can be realized 
by a quantum circuit consisting of two queries to $O_D$ and $O_D^\dagger$ 
and a circuit of depth $\mathcal{O}({\rm poly}(t))$, independent of $N$. 
The circuit uses $t+1$ additional ancilla qubits, which are uncomputed at 
the end of the computation.
\end{theorem}

\begin{proof}
The quantum circuit is depicted in Figure~\ref{fig:diag}.
Let $\varepsilon$ be the required precision for the block encoding and 
let $t = \lceil \log_2(M\pi/c\varepsilon)\rceil + 1$.

The oracle $O_D$,
\[
  O_D\ket{0}_{t+1}\ket{i} 
  = \ket{d_i^{\,\rm sgn}}\Bigl|\tfrac{c|\tilde{d}_{i}|}{M}\Bigr\rangle_{t}\ket{i},
\]
provides $t$-digit approximations $\tilde{d}_i$ of the entries $d_i$ such 
that $|\tilde{d}_i - d_i| \leq 2^{-t}M/c$, or equivalently 
$\left|\frac{c|\tilde{d}_i|}{M} - \frac{c|d_i|}{M}\right| \leq 2^{-t}$.
The extra qubit $\ket{d_i^{\,\rm sgn}}$ encodes the sign of $d_i$ 
($\ket{0}$ for positive, $\ket{1}$ for negative).

Consider the function $f:\{0,1\}^t \to \{0,1\}^t$ where $f(x)$ is the 
output of the Boolean circuit from hypothesis~(ii), implementing a 
$t$-digit approximation of $x \mapsto \frac{1}{\pi}\arcsin(x)$ for 
$x \in [0,c]$. It holds $|f(x) - \frac{1}{\pi}\arcsin(x)| \leq 2^{-t}$.

By the standard reversible simulation of Boolean 
circuits (see Lemma~\ref{lem:10:10}), hypothesis~(ii) implies that there exists 
a unitary $\mathcal{A}$, implementable by a quantum circuit of depth 
$\mathcal{O}({\rm poly}(t))$ and independent of $N$, that computes $f$  
on any superposition of inputs:
\[
  \mathcal{A}\ket{d_i^{\,\rm sgn}}
 \ket{\tfrac{c|\tilde{d}_{i}|}{M}}_{t}\ket{i}
  = \ket{d_i^{\,\rm sgn}}
  \ket{f\left(\tfrac{c|\tilde{d}_{i}|}{M}\right)}_{t}\ket{i}.
\]

The approximation error satisfies:
\begin{eqnarray}
  &&\left|f\!\left(\frac{c|\tilde{d}_{i}|}{M}\right) 
    - \frac{1}{\pi}\arcsin\!\left(\frac{c|d_{i}|}{M}\right)\right| \nonumber\\
  &&\leq\left|f\!\left(\frac{c|\tilde{d}_{i}|}{M}\right)
    -\frac{1}{\pi}\arcsin\!\left(\frac{c|\tilde{d}_{i}|}{M}\right)\right| 
    + \left|\frac{1}{\pi}\arcsin\!\left(\frac{c|\tilde{d}_{i}|}{M}\right)
    -\frac{1}{\pi}\arcsin\!\left(\frac{c|d_{i}|}{M}\right)\right|\nonumber\\
  &&\leq 2^{-t} + c\left|\frac{|\tilde{d}_{i}|}{M} 
    - \frac{|d_{i}|}{M}\right| 
    \leq 2^{1-t} \leq \frac{c\varepsilon}{\pi M} \leq \frac{\varepsilon}{\pi M},
    \label{eq:approx-arcsin}
\end{eqnarray}
where in the second inequality we applied the Lagrange theorem to 
$\frac{1}{\pi}\arcsin(x)$ on $[0,c]$.

Denoting $\tilde{\vartheta}_i = f(c|\tilde{d}_i|/M)$ and applying the 
circuit $U_\vartheta$ from Figure~\ref{fig:thetaf}, we obtain the state
$$
  (-1)^{d_i^{\,\rm sgn}}\!\left(\cos(\pi\tilde{\vartheta}_i)\ket{1}
  + \sin(\pi\tilde{\vartheta}_i)\ket{0}\right)
  \ket{d_i^{\,\rm sgn}}\ket{\tilde{\vartheta}_i}\ket{i}.
$$

By uncomputing the middle register via $\mathcal{A}^\dagger$ followed by 
$O_D^\dagger$, the state becomes
\[
  (-1)^{d_i^{\,\rm sgn}}\!\left(\cos(\pi\tilde{\vartheta}_i)\ket{0}
  + \sin(\pi\tilde{\vartheta}_i)\ket{1}\right)\ket{0}_{t+1}\ket{i}.
\]

From \eqref{eq:approx-arcsin} and the inequality $|\sin(x)-\sin(y)|\leq|x-y|$,
\begin{eqnarray*}
  \left|(-1)^{d_i^{\,\rm sgn}}\sin(\pi\tilde{\vartheta}_i) 
    - (-1)^{d_i^{\,\rm sgn}}\frac{c|d_i|}{M}\right|
  &=& \left|\sin(\pi\tilde{\vartheta}_i) - \frac{c|d_i|}{M}\right|\\
  &\leq& \left|\pi\tilde{\vartheta}_i 
    - \arcsin\!\left(\frac{c|d_i|}{M}\right)\right| 
  < \frac{c\varepsilon}{M},
\end{eqnarray*}
which confirms the $(M/c,\,1,\,\varepsilon)$-block encoding. The encoding 
operator acts as
\begin{equation}\label{eq:oracleD}
  U_D\ket{0}\ket{i} \to \left(\frac{c\tilde{d}_i}{M}
  \ket{0} + \sqrt{1 - \frac{c^2\tilde{d}_i^2}{M^2}}\ket{1}\right)\ket{i}.
\end{equation}

The circuit in Figure~\ref{fig:diag} requires two queries to $O_D$ and 
$O_D^\dagger$, one application each of $\mathcal{A}$, $\mathcal{A}^\dagger$, 
and $U_\vartheta$. By hypothesis~(ii) and the reversible simulation 
argument above, $\mathcal{A}$, $\mathcal{A}^\dagger$, and $U_\vartheta$ 
each have depth $\mathcal{O}({\rm poly}(t))$, with no dependence on $N$. The 
overall circuit depth is therefore $\mathcal{O}({\rm poly}(t))$ plus two queries 
to $O_D$. After uncomputation, the $t+1$ ancilla qubits are returned 
to $\ket{0}_{t+1}$ and can be discarded.
\end{proof}
    
\begin{remark}
The assumption on the Boolean circuit for $\frac{1}{\pi}\arcsin(x)$ is 
reasonable and motivated by the literature. Indeed, classical algorithms for evaluating 
$\arcsin$ to $t$-digit precision, such as Chebyshev polynomial 
approximation evaluated via Horner's rule~\cite{Brent1976}, or CORDIC-type 
iterations, admit Boolean circuits of depth $\mathcal{O}(t^2)$ or $\mathcal{O}(t \log t)$ 
on $t$-bit inputs~\cite{Muller2018}. 

\end{remark}

\begin{figure}
\begin{center}
$
\begin{myqcircuit}
\lstick{\ket{0}} & \qw&\qw	&\qw  &\qw& \qw &\qw&\qw&\qw &\multigate{1}{U_\vartheta} &\qw & \qw	& \qw &\qw &\qw\\
\lstick{\ket{0}}& {/}^{t+1} \qw &\qw&\qw&\multigate{1}{O_D}     & \qw   &\multigate{1}{\mathcal{A}}   &\qw  &\qw   &\ghost{U_\vartheta}   & \qw& \multigate{1}{\mathcal{A}^\dagger} &\qw& \multigate{1}{O_D^\dagger}    &\qw\\
\lstick{\ket{i}}   &{/}^{n} \qw &\qw&\qw& \ghost{O_D} & \qw  &   \ghost{A}   &\qw &\qw&\qw &\qw&\ghost{A^\dagger} &\qw& \ghost{O_D^\dagger} & \qw
\end{myqcircuit}
$
\end{center}
\caption{Circuit for the block encoding of a diagonal matrix}\label{fig:diag}
\end{figure}

The circuit for the block encoding of the inverse of a diagonal matrix is analogous to the encoding of $D$ except for the fact that we need to encode a value proportional to $1/d_{i}$~\cite{Tong2021}.

While Theorem~\ref{diag_encoding} encodes $D$ using an oracle $O_D$ 
that provides $c|d_i|/M$ to $t$-digit precision, encoding $D^{-1}$ 
via the same approach would require a separate oracle for the entries 
$1/d_i$, which may not be available without additional  
preprocessing (with an $\mathcal{O}(N)$ operations in the classical setting). The following theorem shows that a block encoding of 
$D^{-1}$ can instead be realized using an oracle of the same type as 
$O_D$ and avoiding 
any explicit classical computation of the inverse entries.

\begin{theorem}\label{teo:fast_inversion}
Let $D\in \mathbb{R}^{N\times N}$ be an invertible diagonal matrix 
$D= {\rm diag}(d_0,\ldots, d_{N-1})$ and let 
$M=\max_{i=0,\ldots,N-1}|d_{i}|$, $m=\min_{i=0,\ldots,N-1}|d_{i}|$, 
$\mu=M/m$, and $k=c/\mu$ where $c=\frac{\sqrt{\pi^2-1}}{\pi}$ as in 
Theorem~\ref{diag_encoding}.

Let $t = \lceil\log_2(M\pi/\varepsilon)\rceil + 1$ and 
$t' = t + \lceil\log_2(c\mu)\rceil$. Assume that:
\begin{description}
  \item[{\em(i.)}] There exists an oracle $O_D$ such that
    $$
      O_D\ket{0}_{t'+1}\ket{i}
      = \ket{d_i^{\,\rm sgn}}
      \ket{\tfrac{|\tilde{d}_{i}|}{M}}_{t'}\ket{i},
    $$
    where $\ket{d_i^{\,\rm sgn}}$ encodes the sign of $d_i$ and 
    $|\tilde{d}_i/M - d_i/M| \leq 2^{-t'}$.

   \item[{\em(ii.)}] There exists a classical Boolean circuit of depth $\mathcal{O}({\rm poly}(t))$ 
    that computes a $t$-digit fixed-point approximation of 
    $x \mapsto \frac{1}{\pi}{\rm arcsin}(k/x)$ for $x\in[1/\mu, 1]$ 
    given with $t'$-digit precision.
\end{description}

Then a $\left(1/(c\,m),\, 1,\, \varepsilon\right)$-block encoding of $D^{-1}$ can be 
realized by a quantum circuit consisting of two queries to $O_D$ and 
$O_D^\dagger$ and a circuit of depth $\mathcal{O}({\rm poly}(t'))$, independent 
of $N$. The circuit uses $t'+1$ additional ancilla qubits, which are 
uncomputed at the end of the computation.
\end{theorem}
\begin{proof}
The proof follows the same structure as Theorem~\ref{diag_encoding} 
and relies on ideas from~\cite{lin2022lecture}.

Let $\varepsilon$ be the required precision and let 
$t = \lceil\log_2(M\pi/\varepsilon)\rceil + 1$, 
$t' = t + \lceil\log_2(c\mu)\rceil$.

The oracle $O_D$,
$$
  O_D\ket{0}_{t'+1}\ket{i}
  = \ket{d_i^{\,\rm sgn}}
  \ket{\tfrac{|\tilde{d}_{i}|}{M}}_{t'}\ket{i},
$$
provides $t'$-digit approximations of the scaled entries $|d_i|/M$ 
such that $|\tilde{d}_i/M - d_i/M| \leq 2^{-t'}= 2^{-t}/(c\mu)$. 
The extra qubit $\ket{d_i^{\,\rm sgn}}$ encodes the sign of $d_i$ 
($\ket{0}$ for positive, $\ket{1}$ for negative).

Consider the function $g:\{0,1\}^{t'}\to\{0,1\}^{t}$ where $g(x)$ 
is the output of the Boolean circuit from hypothesis~(ii), implementing 
a $t$-digit approximation of $x\mapsto\frac{1}{\pi}\arcsin(k/x)$ for 
$x\in[1/\mu,1]$. It holds $|g(x)-\frac{1}{\pi}\arcsin(k/x)|\leq 2^{-t}$.

By the standard reversible simulation of Boolean 
circuits, hypothesis~(ii) implies that there exists 
a unitary $\mathcal{B}$, implementable by a quantum circuit of depth 
$\mathcal{O}({\rm poly}(t))$ and independent of $N$, that computes $g$  
on any superposition of inputs:
$$
  \mathcal{B}\ket{d_i^{\,\rm sgn}}
  \ket{\tfrac{|\tilde{d}_i|}{M}}_{t'}\ket{i}
  = \ket{d_i^{\,\rm sgn}}
  \ket{g\!\left(\tfrac{|\tilde{d}_i|}{M}\right)}_{t}
  \ket{0}_{\lceil\log_2(c\mu)\rceil}\ket{i}.
$$

The approximation error satisfies:
\begin{eqnarray}
  &&\left|g\!\left(\frac{|\tilde{d}_{i}|}{M}\right) 
    - \frac{1}{\pi}\arcsin\!\left(\frac{kM}{|d_{i}|}\right)\right|
    \nonumber\\
  &&\leq\left|g\!\left(\frac{|\tilde{d}_{i}|}{M}\right)
    -\frac{1}{\pi}\arcsin\!\left(\frac{kM}{|\tilde{d}_{i}|}\right)\right| 
    + \left|\frac{1}{\pi}\arcsin\!\left(\frac{kM}{|\tilde{d}_{i}|}\right)
    -\frac{1}{\pi}\arcsin\!\left(\frac{kM}{|d_{i}|}\right)\right|
    \nonumber\\
  &&\leq 2^{-t} + c\mu\left|\frac{|\tilde{d}_{i}|}{M} 
    - \frac{|d_{i}|}{M}\right| 
  \leq 2^{1-t} \leq \frac{\varepsilon}{\pi M},
  \label{eq:approx-invarcsin}
\end{eqnarray}
where in the second inequality we applied the mean value theorem to 
$\frac{1}{\pi}\arcsin(k/x)$ on $[1/\mu,1]$, using 
$\left|\frac{d}{dx}\frac{1}{\pi}\arcsin\!\left(\frac{k}{x}\right)\right|
\leq \frac{k}{\pi x^2\sqrt{1-k^2/x^2}} \leq c\mu$ 
for $x\in[1/\mu,1]$.

Denoting $\tilde{\omega}_i = g(|\tilde{d}_i|/M)$ and applying the 
circuit $U_\vartheta$ from Figure~\ref{fig:thetaf}, we obtain the state
\[
  (-1)^{d_i^{\,\rm sgn}}\!\left(\sin(\pi\tilde{\omega}_i)\ket{0}
  +\cos(\pi\tilde{\omega}_i)\ket{1}\right)
  \ket{d_i^{\,\rm sgn}}\ket{\tilde{\omega}_i}_t \ket{0}_{\lceil\log_2(c\mu)\rceil}\ket{i}.
\]

By uncomputing the middle register via $\mathcal{B}^\dagger$ followed 
by $O_D^\dagger$, the state becomes
\[
  (-1)^{d_i^{\,\rm sgn}}\!\left(\sin(\pi\tilde{\omega}_i)\ket{0}
  +\cos(\pi\tilde{\omega}_i)\ket{1}\right)\ket{0}_{t'+1}\ket{i}.
\]

From \eqref{eq:approx-invarcsin} and the inequality 
$|\sin(x)-\sin(y)|\leq|x-y|$, we obtain
\begin{eqnarray*}
  \left|(-1)^{d_i^{\,\rm sgn}}\sin(\pi\tilde{\omega}_i) 
    - (-1)^{d_i^{\,\rm sgn}}\frac{Mk}{|d_i|}\right|
  &=& \left|\sin(\pi\tilde{\omega}_i) - \frac{Mk}{|d_i|}\right|\\
  &\leq& \left|\pi\tilde{\omega}_i 
    - \arcsin\!\left(\frac{Mk}{|d_i|}\right)\right| 
  \leq \frac{\varepsilon}{M},
\end{eqnarray*}
which confirms the $(1/(c\, m),\,1,\,\varepsilon)$-block encoding of 
$D^{-1}$, since $M\,k = c \,M/\mu = c\,m$. More precisely, the encoding operator 
acts as
\begin{equation}\label{eq:oracleDinv}
  V_D\ket{0}\ket{i} \to 
  \left(\frac{Mk}{\tilde{d}_i}\ket{0}
  + \sqrt{1-\frac{M^2 k^2}{\tilde{d}_i^2}}\ket{1}\right)\ket{i}.
\end{equation}

The circuit in Figure~\ref{fig:fast_inv} requires two queries to 
$O_D$ and $O_D^\dagger$, one application each of $\mathcal{B}$, 
$\mathcal{B}^\dagger$, and $U_\vartheta$. By hypothesis~(ii) and 
the reversible simulation argument above, $\mathcal{B}$, 
$\mathcal{B}^\dagger$, and $U_\vartheta$ each have depth 
$\mathcal{O}({\rm poly}(t'))$.

The overall circuit 
depth is therefore $\mathcal{O}({\rm poly}(t'))$ plus two queries to $O_D$. 
After uncomputation, the $t'+1$ ancilla qubits are returned to 
$\ket{0}_{t'+1}$ and can be discarded.
\end{proof}

\begin{figure}
\begin{center}
$
\begin{myqcircuit}
\lstick{\ket{0}} & \qw&\qw	&\qw  &\qw& \qw &\qw&\qw&\qw &\multigate{1}{U_\vartheta} &\qw & \qw	& \qw &\qw &\qw\\
\lstick{\ket{0}}& {/}^{t'+1} \qw &\qw&\qw&\multigate{1}{O_D}     & \qw   &\multigate{1}{\mathcal{B}}   &\qw  &\qw   &\ghost{U_\vartheta}   & \qw& \multigate{1}{\mathcal{B}^\dagger} &\qw& \multigate{1}{O_D^\dagger}    &\qw\\
\lstick{\ket{i}}   &{/}^{n} \qw &\qw&\qw& \ghost{O_D} & \qw  &   \ghost{\mathcal{B}}   &\qw &\qw&\qw &\qw&\ghost{\mathcal{B}^\dagger} &\qw& \ghost{O_D^\dagger} & \qw
\end{myqcircuit}
$
\end{center}
\caption{Circuit for the fast inversion of a Diagonal Matrix}\label{fig:fast_inv}
\end{figure}

\begin{remark} \label{parag:oracles}
The oracle $O_D$ can be realized in several ways. 
If the entries  are stored in a quantum random-access memory (QRAM)~\cite{giovannetti2008quantum, kerenidis2017quantum}, 
each query to $O_D$ costs $\mathcal{O}({\rm polylog}(N))$ gates, possibly at the expense  of some classical preprocessing step to load the data.

Alternatively, one can realize the oracle with a state-preparation circuit. In several cases such circuits can be realized with depth $\mathcal{O}({\rm polylog}(N))$, for instance when $d_i=g(i)$ and $g(x)$ is a smooth bounded arithmetic function~\cite[Proposition 4]{li2023efficient}.

\end{remark}

\section{QBE for one-pair matrices}
\label{sec:building_blocks}
In this section we present the block encoding for one-pair  matrices. The block encoding is based on the factorization \eqref{eq:mainfact}, and the main ingredient is the block encoding of the lower triangular factor $L$.

\subsection{QBE of matrix $L$}  \label{sec:QBEL}
The lower triangular matrix $L$ in \eqref{eq:mainfact} is Toeplitz and can therefore be block-encoded as described in \cite{wan2021block}. However, given its additional structure, we can devise a suitable ad-hoc encoding. 
Recall that the $N\times N$ downshift matrix $$Z_N=\begin{bmatrix}
     &&& 1 \\ 1 &&& \\ & \ddots && \\ && 1 & 
\end{bmatrix}$$ 
can be encoded by the quantum circuit described in~\cite{camps2024explicit}
using a cascade of Toffoli gates. For example, the case $N=8$ is obtained with the following circuit 
\begin{center}
$
\begin{myqcircuit}
\lstick{q_{0}}	&	\ctrl{1}	&	\ctrl{1}	&	\gate{X}	&	\qw	\\
\lstick{q_{1}}	&	\ctrl{1}	&	\targ	&	\qw	&	\qw	\\
\lstick{q_{2}}	&	\targ	&	\qw	&	\qw	&	\qw	\\
\end{myqcircuit}
$
\end{center}
which performs the operation $x\rightarrow x+1\, {\rm mod}\, N$ on $x=\sum_{k=0}^{n-1}2^{j_k}$. For more details, see, e.g., \cite{camps2024explicit} and the {\tt leftshift} function in the QCLAB\footnote{https://github.com/QuantumComputingLab/qclab} toolbox for MATLAB.

In the computational basis, one can build each column of $L$ by applying the {\tt leftshift} operation a suitable number of times to a vector of length $2N$ of the form 
$$[\underbrace{1/N,\dots,1/N}_{N},\underbrace{0,\dots,0}_{N}]^T,$$
where the tail zeros are added to generate the lower triangular form. 
The above vector is readily obtained via a Hadamard transform on the first $n$ qubits.

To cyclically shift the columns we need to apply the following unitary block diagonal matrix
\begin{equation} \label{Z-block}
    Z_{\rm block}=\begin{bmatrix}
        I_{2N}& 0 & 0 & \cdots& 0\\
        0 &Z_{2N} & 0 & \cdots & 0\\
        0 & 0 & Z_{2N}^2 & \cdots  &0\\
        \vdots & & & \ddots & \vdots \\
        0 & 0 & 0 & \cdots & Z_{2N}^{N-1}
    \end{bmatrix}.
\end{equation}
Indeed once we have a circuit for~\eqref{Z-block} we can get the block encoding of $L$ with the following steps
\begin{equation} \label{UL}
U_L=  (H^{\otimes n}\otimes I^{\otimes (n+1)}) Z_{\rm block} (H^{\otimes n}\otimes I^{\otimes (n+1)}).     
\end{equation} 
To realize the circuit which describes~\eqref{Z-block} with a polylogarithmic cost in $N$, we note that $Z_{\rm block}$ can be obtained with only $n$ controlled leftshift operations. Indeed, for example, with $n=3, N=8$ we have

$$
Z_{\rm block}=\left(\begin{bmatrix}
    I&\\ &Z^4_{2N}
\end{bmatrix}\otimes I_4\right)\left(I_2\otimes \begin{bmatrix}
    I&\\ &Z^2_{2N}
\end{bmatrix}\otimes I_2\right)\left(I_4\otimes \begin{bmatrix} 
I&\\&Z_{2N}
\end{bmatrix}\right)
$$
and we know that 
$Z_{2k}^2=Z_k\otimes I_2$. Hence, 
$Z_{2N}^2= Z_N\otimes I_2$, and $Z_{2N}^4=(Z_N^2 \otimes I_2)=(Z_{N/2}\otimes I_4).$ 
The extended form of~\eqref{UL} is as follows
$$
U_L =
\begin{pNiceArray}{cccc|ccc}
    \frac{1}{N} & 0 & 0 & 0 & * &\cdots& * \\
    \frac{1}{N} & \ddots & 0 & 0 & * &\vdots & * \\
    \vdots & \vdots& \frac{1}{N} & 0 & \vdots &\ddots& \vdots \\
    \frac{1}{N} & \frac{1}{N} & \cdots & \frac{1}{N} & * &\cdots & * \\
\hline
    * & * &  \cdots & * &* &\cdots & * \\
    \vdots &\vdots& \ddots & \vdots & \vdots&\ddots&\vdots \\
    * & * &  \cdots & * &* &\cdots & * \\
\end{pNiceArray}
$$
which has the matrix $\frac{1}{N} L$ in the top-left corner. Note that $U_L$ has size $N^2$ while matrix $L$ has size $N$.

\begin{figure}
\centering
\begin{tikzpicture}
\node[inner sep=0] (circuit) {$
\begin{myqcircuit}
\lstick{q_0} & \qw & \gate{H}& \qw  & \qw   & \qw& \qw  & \cds{2}{\cdots}& \qw  & \ctrl{4} & \qw  & \gate{H}& \qw\\
\vdots &&&&&&& &&&&&\\
\lstick{q_{n-2}}& \qw & \gate{H}& \qw & \qw & \qw  & \ctrl{2} & \qw & \qw & \qw& \qw & \gate{H}& \qw\\
\lstick{q_{n-1}}& \qw & \gate{H}& \qw & \ctrl{1} & \qw   & \qw& \qw & \qw & \qw& \qw  & \gate{H}& \qw\\
\lstick{q_n}& \qw & \qw & \qw  &\multigate{4}{Z_{2^{n+1}}}  & \qw  &\multigate{3}{Z_{2^{n}}} & \qw & \qw  & \multigate{1}{Z_{2^{2}}}& \qw & \qw & \qw\\
\lstick{q_{n+1}}& \qw & \qw & \qw &\ghost{Z_{2^{n+1}}}  & \qw  &\ghost{Z_{2^{n}}} &  \cds{2}{\cdots} & \qw   &\ghost{Z_{2^{2}}}& \qw& \qw & \qw \\
\vdots &&&& &&&& &&&&\\
\lstick{q_{2n-1}} & \qw& \qw  & \qw &\ghost{Z_{2^{n+1}}}  & \qw &\ghost{Z_{2^{n}}} & \qw & \qw & \qw &\qw& \qw & \qw\\
\lstick{q_{2n}} & \qw& \qw  & \qw &\ghost{Z_{2^{n+1}}}  & \qw & \qw & \qw & \qw & \qw & \qw& \qw & \qw
\end{myqcircuit}
$};
\draw[dashed, thick, rounded corners=3pt]
  ([shift={(3em, 0.6em)}] circuit.north west) rectangle ([shift={(-2.5em,-0.6em)}] circuit.south east);
\end{tikzpicture}
\caption{Circuit realizing $U_L$, the $(N, \log(N)+1, 0)-$ block encoding of $L$.}
\label{fig:LCirc}
\end{figure}
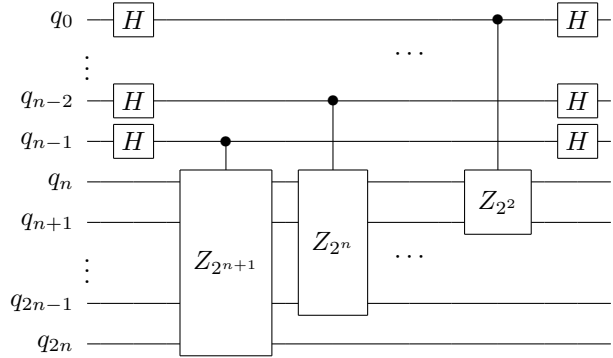

\begin{theorem}\label{theo:Ltheorem}
    Given $L$,an $N\times N$ lower triangular matrix of all 1's, with $N=2^n$, the circuit described in  Figure~\ref{fig:LCirc} produces a $(N, \log(N)+1, 0)-$block encoding of $L$. The circuit uses $2\log(N)+1$ qubits and has a depth of $\mathcal{O}({\rm polylog}(N))$.
\end{theorem}
\begin{proof}
Starting from the initial configuration  $\ket{\psi_0}=\ket{0}^{\otimes n}\ket{0}\ket{j},$ $0\le j\le N-1$ we transition to state 
$\ket{\psi_1}=\frac{1}{\sqrt{N}}\left( \sum_{k=0}^{N-1}\ket{k} \right)\ket{0}\ket{j}$ after the application of the Hadamard gates to the first $n$ qubits.
The controlled operations are represented by matrix $Z_{\rm block}$. After its application we get 
$$ \ket{\psi_2}=\frac{1}{\sqrt{N}}\left( \sum_{k=0}^{N-1}\ket{k} \ket{k+j}\right).
$$
The effect of the last sequence of Hadamard gates on the first $n$ qubits is such that
$$
\ket{\psi_3}= \frac{1}{\sqrt{N}}\sum_{k=0}^{N-1} H^{\otimes n}\ket{k} \ket{k+j}.
$$
Projecting along the direction of the $\ket{0}^{\otimes n}\ket{0}\ket{i}$, with $i=j+1, \ldots, N-1$, we get $1/N$, which is $L_{ij}/N$, whereas we obtain $0$ if $i<j$:
$$
\bra{i} \bra{0} \bra{0}^{\otimes n}\ket{\psi_3}= \frac{1}{N} \sum_{k=0}^{N-1}\bra{i} \bra{0}\ket{k+j}.
$$
Indeed $\bra{0}^{\otimes n} H^{\otimes n}\ket{k}= \frac{1}{\sqrt{N}}$
and $\sum_{k=0}^{N-1}\bra{i} \bra{0}\ket{k+j}=1$ if $i\ge j$, and $0$ if $i<j$.
This proves that the circuit in Figure~\ref{fig:LCirc} is an $(N, \log(N)+1, 0)-$block encoding of $L$. 
The depth of the circuit is polylogarithmic in $N$ since $Z_{\rm block}$ can be obtained with $n$ controlled leftshift operations. Each controlled $Z_{2^k}$ gate  costs $\mathcal{O}(k)$ operations.
    
\end{proof} 
\subsection{QBE of matrix $D_u$} 
From Proposition \ref{prop:decomoposition} we know that $D_u$ is a diagonal matrix whose diagonal entries are the elements of the vector $u$. The block encoding of $D_u$ is obtained by applying Theorem \ref{diag_encoding} to the vector $u$, which produces a $(\max_i{|u_i|}/c,1,\varepsilon_u)$  block encoding of $D_u$, where $c=\frac{\sqrt{\pi^2-1}}{\pi}$.

Note that we can reduce the error of the block encoding of $D_u$ by using more digits in the representation of $u$. As underlined in Theorem~\ref{diag_encoding}, the $t+1$ qubits used for representing each entry of $u$ with precision $2^{-t}$ can then be released at the end of the computation. The cost of this encoding is polynomial in $t$ plus two oracle calls to $O_D$.

\subsection{QBE of matrix $\Delta_z$} 
\label{sec:deltaz}

The matrix $ \Delta_z $ is a diagonal matrix where each diagonal entry corresponds to the difference between consecutive elements of the vector $ z $, except for the first entry, which is simply $ z_0 $. It can be expressed as follows:

$$\Delta_z = \begin{bmatrix}  z_0 & & & &\\ & z_1 - z_0 & & &\\ & & \ddots & &\\ & & & z_{i} - z_{i-1} & \\ & & & & \ddots & \\ & & & & & z_{N-1} - z_{N-2} \end{bmatrix}, \quad z_i = \frac{v_i}{u_i}.$$

Since the elements of $ z $ are given by $ v_i / u_i $, we can obtain the block encoding of $ \Delta_z $ using the two vectors $ u $ and $ v $. 
There are several ways of encoding $\Delta_z$. For example, one can encode directly the diagonal matrix with entries $z_i-z_{i+1}=(u_iv_{i+1}-u_{i+1}v_i)/(v_i v_{i+1})$; alternatively, one may take as input the vector $z$ and encode $\Delta_z$ from $D_z$ and a shifted version of $D_z$.

Here, we assume that in input we only have the two generators $u$ and $v$ already stored in a quantum system, for example in a QRAM, ensuring the possibility of using the data multiple times. This approach keeps all computations entirely within the quantum domain. We proceed as follows.

The first step involves block-encoding the diagonal matrix $\text{diag}(z_0, z_1, \ldots, z_{N-1})$ as the product  of the two diagonal matrices,  $ D_v $ encoded exploiting Theorem~\ref{diag_encoding} and $ D_u^{-1} $ encoded exploiting Theorem~\ref{teo:fast_inversion}. 

Next, we need to block-encode $\text{diag}(z_{\downarrow})$, where $ z_{\downarrow} = [0, z_0, \ldots, z_{N-2}]^T $ is a shifted version of $ z $ that includes a zero as the first element and omits $ z_{N-1} $. One of the possibility to encode the diagonal matrix $\text{diag}(z_{\downarrow})$ is to block-encode $ D_{v_{\downarrow}} = \text{diag}(0, v_0, v_1, \ldots, \\v_{N-2}) $ and $ D_{u_{\downarrow}}^{-1} = (\text{diag}(u_{N-1}, u_0, u_1, \ldots, u_{N-2}))^{-1} $, and multiply the two matrices.

Finally, we combine the block encodings of $\text{diag}(z_0, z_1, \ldots, z_{N-1})$ and $\text{diag}(0, z_0, z_1,\\ \ldots, z_{N-2})$ by introducing an ancillary qubit and performing an LCU operation to compute the difference between the two block encodings.

\begin{figure}
\centering
$
\begin{myqcircuit}
\lstick{\ket{0}}	& \gate{H}          & \ctrlo{2}       & \qw     & \ctrlo{1}        & \qw	& \ctrlo{2} 	               & \qw& \qw	&  \qw      &   \qw  & \qw & \qw & \ctrl{2}       & \qw     & \ctrl{1}        & \qw	& \ctrl{2} 	               & \qw	&  \gate{H}    &   \qw  & \gate{X} & \qw  \\
\lstick{\ket{0}}  & \qw	& \qw	                & \qw      & \qswap		   & \qw	& \qw & \qw	&  \qw 	&  \qw      &   \qw  & \qw                 & \qw & \qw	                & \qw      & \qswap		   & \qw	& \qw & \qw	&  \qw      &   \qw  & \qw                 & \qw \\
\lstick{\ket{0}}   & \qw	& \multigate{1}{U_{D_{v}}}    & \qw	   & \qswap \qwx   & \qw	& \multigate{1}{{U_{D_{u}^{-1}}}}   & \qw	&  \qw   &   \qw  &\qw  & \qw & \qw	& \multigate{1}{U_{D{v}^{\downarrow}}}    & \qw	   & \qswap \qwx   & \qw	& \multigate{1}{{U_{D_{u^{\downarrow}}^{-1}}}}   & \qw	&  \qw   &   \qw  &\qw  & \qw  \\
\lstick{\ket{i}}	  &  \qw {/}^{n} & \ghost{{U_{D{v}}}}           & \qw	   & \qw           & \qw	& \ghost{{U_{D_{u}^{-1}}}} 	       & \qw	&  \qw         &   \qw  & \qw         & \qw & \qw& \ghost{{U_{D_{v^{\downarrow}}}}}           & \qw	   & \qw           & \qw	& \ghost{{U_{D_{u^{\downarrow}}^{-1v}}}} 	       & \qw	&  \qw         &   \qw  & \qw         & \qw \\
\end{myqcircuit}
$
\caption{Circuit for the block encoding of the diagonal matrix $\Delta_z$.}\label{fig:BEDz}
\end{figure}

The circuit realizing the block encoding of $\Delta_z$ is depicted in Figure~\ref{fig:BEDz}. We have the following result.
\begin{theorem} \label{lemma:Dz}
Let $u$ and $v$ be two $n$-qubit vectors. Let $M_v=\max_{i=0,\ldots N-1} |v_i|$, $M_u=\max_{i=0,\ldots N-1} |u_i|$, $m_u=\min_{i=0, \ldots, N-1} |u_i|\neq 0$ and $c= \frac{\sqrt{\pi^2-1}}{\pi}$. \\
Then the circuit in Figure~\ref{fig:BEDz} realizes a $\left(2\frac{M_v}{c^2 m_u}, 3, 2(\frac{1}{m_u} \epsilon_v+ M_v\epsilon_u)\right)-$block encoding of $\Delta_z$ defined in Proposition~\ref{prop:decomoposition}, where $\epsilon_u$ and $\epsilon_v$ are such that
$$
\left |v_i -\frac{M_v}{c}U_{D_v}(i,i) \right|\le \epsilon_v, \ \ \left |\frac{1}{u_i}-\frac{c}{m_u}U_{D_u^{-1}}(i,i) \right|\le \epsilon_u.$$ The circuit requires $n+3$ qubits and has depth $\mathcal{O}({\rm{poly}}(t))$ plus 8 calls to the diagonal oracles $O_D$ where 
$$t=\max\left\{ \left\lceil \log_2\left(\frac{M_v \pi}{c \epsilon_v}\right) \right\rceil , \left\lceil \log_2\left( \frac{\pi \, M_u}{\epsilon_u }\right) \right\rceil +\left\lceil \log_2\left( \frac{c\,M_u}{ m_u}\right) \right\rceil\right\}+1.
$$
\end{theorem}

\begin{proof}    
Here we prove the correctness of the circuit in Figure~\ref{fig:BEDz}, in particular, we want to show that starting with the initial configuration 
$\ket{\phi_0}=\ket{0}\ket{0}\ket{0}\ket{i}$ we end up with a state $\ket{\phi_F}= U_{\Delta_z} \ket{000}\ket{i}$ such that $\bra{000}\bra{i}\ket{\phi_F}=\lambda\left(z_i-z_{i-1}\right),$ for a known $\lambda$ and $i>0$, whereas for $i=0$ we expect $\bra{000}\bra{0}\ket{\phi_F}=\lambda\left(z_0-0\right)$.  In the following, as before let $c=\frac{\sqrt{\pi^2-1}}{\pi}$.
After the application of the Hadamard gate on the first qubit we have the state
$$
  \frac{1}{\sqrt{2}}\ket{0}\ket{0}\ket{0}\ket{i}+\frac{1}{\sqrt{2}}\ket{1}\ket{0}\ket{0}\ket{i}.\\
$$
Then on the branch where the first qubit is zero we apply $U_{D_v}$ as described by equation~\eqref{eq:oracleD}. We get
$$
   \frac{1}{\sqrt{2}}\ket{0}\ket{0}\left( \frac{c\,v_i}{M_v}\ket{0}+\sqrt{1-\frac{c^2 v_i^2}{M_v^2}}\ket{1}\right)\ket{i}+\frac{1}{\sqrt{2}}\ket{1}\ket{0}\ket{0}\ket{i}.
$$ 
After the controlled swap between the second and the third qubit,  and the conditional application of $U_{D_u^{-1}}$ according to~\eqref{eq:oracleDinv}, we get

  \begin{eqnarray*}
       &\frac{1}{\sqrt{2}}\ket{0}\left( \frac{c\,v_i}{M_v}\ket{0}+\sqrt{1-\frac{c^2 v_i^2}{M_v^2}}\ket{1}\right)\left(\frac{c\,m_u}{u_i}\ket{0} +
       \sqrt{1-\frac{c^2\,m_u^2}{u_i^2}} \ket{1} \right)\ket{i}+\\ &+\frac{1}{\sqrt{2}}\ket{1}\ket{0}\ket{0}\ket{i}.
  \end{eqnarray*}
  
Similarly on the branch where the first qubit is 1 we apply $U_{D_{v_\downarrow}}$ followed by a controlled swap and  $U_{D_{u_\downarrow}^{-1}}.$ At the end we have
  \begin{eqnarray*}
 &\frac{1}{\sqrt{2}} \ket{0}\left( \frac{c^2 m_u v_i}{M_v u_i} \ket{00}\ket{i}+ \sum_{k=1}^3 \lambda_k \ket{k}\ket{i}\right)+\\
 & \frac{1}{\sqrt{2}} \ket{1}\left( \frac{c^2 m_u v_{i-1}}{M_v u_{i-1}} \ket{00}\ket{i}+ \sum_{k=1}^3 \hat\lambda_k \ket{k}\ket{i}\right),
  \end{eqnarray*}
  where $\lambda_k$ and $\hat \lambda_k$, for $k=1,\, 2,\, 3$ are the combinations of the different coefficients of $\ket{0}$ and $\ket{1}$ in the second and third register; for example it holds $\lambda_1= \frac{c v_i}{M_v} \sqrt{1- \frac{c^2\,m_u^2}{u_i^2}}.$ Note that, for $i=0$, we have $v_{i-1}/u_{i-1}=0$.

Applying  $H$ and $X$ on the first qubit we get
\begin{eqnarray*}
&\ket{\phi_F}=\frac{1}{2}(\ket{0}+\ket{1})\left( \frac{c^2 m_u v_i}{M_v u_i} \ket{00}\ket{i}+ \sum_{k=1}^3 \lambda_k \ket{k}\ket{i}\right)- \\
&-
\frac{1}{2}(\ket{0}-\ket{1})\left( \frac{c^2 m_u v_{i-1}}{M_v u_{i-1}} \ket{00}\ket{i}+ \sum_{k=1}^3 \hat\lambda_k \ket{k}\ket{i}\right)=\\
&= \frac{c^2}{2} \frac{m_u}{M_v}\left( \frac{v_i}{u_i} - \frac{v_{i-1}}{u_{i-1}} \right)\ket{000}\ket{i}
+\sum_{k=1}^8 \delta_k \ket{k}\ket{i},
 \end{eqnarray*}
where $\delta_k=\lambda_k-\hat \lambda_k$, for $k=1,\,2,\,3$, $\delta_4= \frac{c^2}{2} \frac{m_u}{M_v}\left( \frac{v_i}{u_i} + \frac{v_{i-1}}{u_{i-1}} \right)$ and $\delta_{k}=\lambda_{(k-1)/4}+\hat \lambda_{(k-1)/4}$, for $k=5,\, 6,\, 7$.   

So we get that 
$$
\bra{000}\langle i|U_{\Delta_z}\phi_0\rangle= \frac{c^2}{2}\frac{m_u}{M_v}\left(z_i-z_{i-1}\right), \quad i>0,
$$ 
and  $\bra{000}\langle 0 \ket{\phi_F}= \frac{c^2}{2}\frac{m_u}{M_v}z_0.$

This proves that the circuit in Figure~\ref{fig:BEDz} realizes a  block encoding of matrix $\Delta_z$ with a scaling factor $2 \frac{M_v}{c^2m_u}$. 
The error for the block encoding results from the following chain of inequalities. Given that $|v_i-\frac{M_v}{c} U_{D_v}(i,i)|\le \epsilon_v$ and $| \frac{1}{u_i}-\frac{1}{c\,m_u} U_{D_u^{-1}}(i,i)|\le \epsilon_u,$ we have
\begin{eqnarray*}
  && \left|\frac{v_i}{u_i}-\frac{M_v}{c^2 m_u} U_{D_v}(i,i)U_{D_u^{-1}}(i,i)\right|\le \\
   &\le&\left|\frac{1}{u_i}\right|\left|v_i- \frac{M_v}{c} U_{D_v}(i,i) \right|+\frac{M_v}{c} |U_{D_v}(i, i)|\left| \frac{1}{u_i}- \frac{1}{c\, m_u}U_{D_u^{-1}}(i,i)\right| \\
   &\le& \frac{1}{m_u} \epsilon_v+ M_v\epsilon_u
\end{eqnarray*}
We can achieve a similar bound for $z_{i-1}=\frac{v_{i-1}}{u_{i-1}}$. Combining the two encodings we get that the error is bounded by $2(\frac{1}{m_u} \epsilon_v+ M_v\epsilon_u)$ as claimed.
The depth of the circuit is obtained using Theorem~\ref{diag_encoding} and~\ref{teo:fast_inversion} for $v$ and $u$ respectively.
 
\end{proof}

\subsection{The final block encoding}
\label{sec:finalBE}
Now we present the final block-encoding. We obtain it by combining the block-encodings of the matrices $U_{D_u}$, $U_L$ (and its transpose), and $U_{\Delta_z}$ introduced in Section~\ref{sec:building_blocks}, following the decomposition from Proposition~\ref{prop:decomoposition}. From here, we can get the whole block-encoding of the one-pair semiseparable matrix $S$ with the circuit in Figure~\ref{fig:BE-S}, shown in more detail for $N=8$ in Figure~\ref{fig:BE-Sexample}.

\begin{figure}
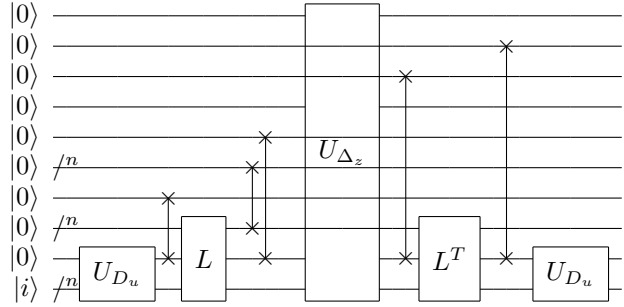

\begin{center}
$
\begin{myqcircuit}
 \lstick{\ket{0}} & \qw & \qw & \qw &  \qw &\qw &\qw&\qw & \qw &\qw &  \multigate{9}{U_{\Delta_z}} & \qw & \qw &\qw &\qw &\qw &\qw &\qw &\qw &\\
  \lstick{\ket{0}} & \qw & \qw & \qw &  \qw &\qw &\qw&\qw & \qw &\qw &  \qw & \qw & \qw &\qw & \qw & \qswap &\qw &\qw &\qw &\\
    \lstick{\ket{0}} & \qw & \qw & \qw &  \qw &\qw &\qw&\qw & \qw &\qw &  \qw & \qw & \qswap &\qw & \qw & \qw \qwx &\qw &\qw &\qw &\\
 \lstick{\ket{0}}& \qw  & \qw & \qw & \qw &\qw &\qw&\qw &\qw &\qw & \ghost{U_{\Delta_z}}& \qw &\qw\qwx  &\qw &\qw &\qwx \qw &\qw &\qw &\qw &\\
 \lstick{\ket{0}}& \qw &  \qw &\qw & \qw &\qw & \qw& \qswap&\qw &\qw &\qw &\qw  &\qw \qwx &\qw &\qw &\qwx \qw &\qw &\qw &\qw &\\
  \lstick{\ket{0}}& {/}^n\qw &  \qw &\qw & \qw &\qw & \qswap& \qw\qwx&\qw &\qw &\qw &\qw  &\qw \qwx&\qw &\qw &\qwx \qw &\qw &\qw &\qw &\\
 \lstick{\ket{0}} & \qw & \qw & \qswap &  \qw &\qw &\qw \qwx& \qw \qwx & \qw &\qw &  \qw & \qw  & \qw\qwx  &\qw &\qw & \qw\qwx &\qw &\qw &\qw &\\
  \lstick{\ket{0}}& {/}^n \qw &  \qw &\qw \qwx & \multigate{2}{L} & \qw&\qswap \qwx & \qw\qwx&\qw &\qw  &\qw &\qw  &\qw\qwx &\multigate{2}{L^T} &\qw &\qwx \qw &\qw &\qw &\qw &\\
 \lstick{\ket{0}}& \qw  & \multigate{1}{U_{D_u}} &\qswap \qwx  & \ghost{L}&\qw&\qw &  \qswap \qwx & \qw &\qw & \ghost{U_{\Delta_z}} &\qw  &\qswap \qwx &\ghost{L^T}&\qw & \qswap \qwx&\qw &\multigate{1}{U_{D_u}} &\qw \\
 \lstick{\ket{i}}& {/}^n \qw  & \ghost{U_{D_u}}  & \qw & \ghost{L}&\qw &\qw&\qw &\qw &\qw & \ghost{U_{\Delta_z}}&\qw &\qw &\ghost{L^T}&\qw &\qw &\qw &\ghost{U_{D_u}}  & \qw &\\
\end{myqcircuit}
$

\end{center}
\caption{Final scheme of the circuit for the block encoding of $S$. The swap gates are necessary to correctly perform the multiplication of the different factors.}
\label{fig:BE-S}
\end{figure}

\begin{figure}
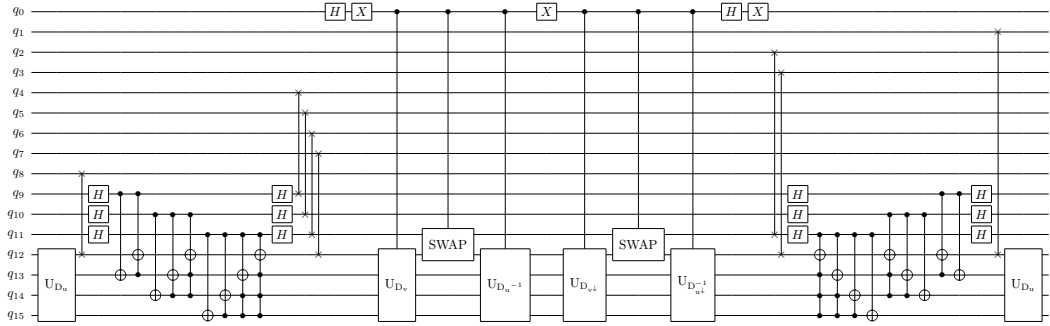

\begin{center}
\scalebox{0.5}{
$
\begin{myqcircuit}
\lstick{q_{0}}  &       \qw     &       \qw     &       \qw     &       \qw     &       \qw     &       \qw     &       \qw     &       \qw     &       \qw     &       \qw     &       \qw     &       \qw     & \qw      &       \qw     &       \qw     &       \qw     &       \qw     &       \gate{H}        &       \gate{X}        &       \ctrl{12}       &       \ctrl{11}       &       \ctrl{12}       &       \gate{X}  &\ctrl{12}       &       \ctrl{11}       &       \ctrl{12}       &       \gate{H}        &       \gate{X}        &       \qw     &       \qw     &       \qw     &       \qw     &       \qw     &       \qw     & \qw      &       \qw     &       \qw     &       \qw     &       \qw     &       \qw     &       \qw     &       \qw     &       \qw     &       \qw     \\
\lstick{q_{1}}  &       \qw     &       \qw     &       \qw     &       \qw     &       \qw     &       \qw     &       \qw     &       \qw     &       \qw     &       \qw     &       \qw     &       \qw     & \qw      &       \qw     &       \qw     &       \qw     &       \qw     &       \qw     &       \qw     &       \qw     &       \qw     &       \qw     &       \qw     &       \qw     &       \qw     &       \qw&       \qw     &       \qw     &       \qw     &       \qw     &       \qw     &       \qw     &       \qw     &       \qw     &       \qw     &       \qw     &       \qw     &       \qw     &       \qw     & \qw      &       \qw     &       \qswap\qwx[11]  &       \qw     &       \qw     \\
\lstick{q_{2}}  &       \qw     &       \qw     &       \qw     &       \qw     &       \qw     &       \qw     &       \qw     &       \qw     &       \qw     &       \qw     &       \qw     &       \qw     & \qw      &       \qw     &       \qw     &       \qw     &       \qw     &       \qw     &       \qw     &       \qw     &       \qw     &       \qw     &       \qw     &       \qw     &       \qw     &       \qw&       \qw     &       \qw     &       \qswap\qwx[9]   &       \qw     &       \qw     &       \qw     &       \qw     &       \qw     &       \qw     &       \qw     &       \qw     &       \qw     &       \qw&       \qw     &       \qw     &       \qw     &       \qw     &       \qw     \\
\lstick{q_{3}}  &       \qw     &       \qw     &       \qw     &       \qw     &       \qw     &       \qw     &       \qw     &       \qw     &       \qw     &       \qw     &       \qw     &       \qw     & \qw      &       \qw     &       \qw     &       \qw     &       \qw     &       \qw     &       \qw     &       \qw     &       \qw     &       \qw     &       \qw     &       \qw     &       \qw     &       \qw&       \qw     &       \qw     &       \qw     &       \qswap\qwx[9]   &       \qw     &       \qw     &       \qw     &       \qw     &       \qw     &       \qw     &       \qw     &       \qw     &       \qw&       \qw     &       \qw     &       \qw     &       \qw     &       \qw     \\
\lstick{q_{4}}  &       \qw     &       \qw     &       \qw     &       \qw     &       \qw     &       \qw     &       \qw     &       \qw     &       \qw     &       \qw     &       \qw     &       \qw     & \qw      &       \qswap\qwx[5]   &       \qw     &       \qw     &       \qw     &       \qw     &       \qw     &       \qw     &       \qw     &       \qw     &       \qw     &       \qw     &       \qw     & \qw      &       \qw     &       \qw     &       \qw     &       \qw     &       \qw     &       \qw     &       \qw     &       \qw     &       \qw     &       \qw     &       \qw     &       \qw     &       \qw&       \qw     &       \qw     &       \qw     &       \qw     &       \qw     \\
\lstick{q_{5}}  &       \qw     &       \qw     &       \qw     &       \qw     &       \qw     &       \qw     &       \qw     &       \qw     &       \qw     &       \qw     &       \qw     &       \qw     & \qw      &       \qw     &       \qswap\qwx[5]   &       \qw     &       \qw     &       \qw     &       \qw     &       \qw     &       \qw     &       \qw     &       \qw     &       \qw     &       \qw     & \qw      &       \qw     &       \qw     &       \qw     &       \qw     &       \qw     &       \qw     &       \qw     &       \qw     &       \qw     &       \qw     &       \qw     &       \qw     &       \qw&       \qw     &       \qw     &       \qw     &       \qw     &       \qw     \\
\lstick{q_{6}}  &       \qw     &       \qw     &       \qw     &       \qw     &       \qw     &       \qw     &       \qw     &       \qw     &       \qw     &       \qw     &       \qw     &       \qw     & \qw      &       \qw     &       \qw     &       \qswap\qwx[5]   &       \qw     &       \qw     &       \qw     &       \qw     &       \qw     &       \qw     &       \qw     &       \qw     &       \qw     & \qw      &       \qw     &       \qw     &       \qw     &       \qw     &       \qw     &       \qw     &       \qw     &       \qw     &       \qw     &       \qw     &       \qw     &       \qw     &       \qw&       \qw     &       \qw     &       \qw     &       \qw     &       \qw     \\
\lstick{q_{7}}  &       \qw     &       \qw     &       \qw     &       \qw     &       \qw     &       \qw     &       \qw     &       \qw     &       \qw     &       \qw     &       \qw     &       \qw     & \qw      &       \qw     &       \qw     &       \qw     &       \qswap\qwx[5]   &       \qw     &       \qw     &       \qw     &       \qw     &       \qw     &       \qw     &       \qw     &       \qw     & \qw      &       \qw     &       \qw     &       \qw     &       \qw     &       \qw     &       \qw     &       \qw     &       \qw     &       \qw     &       \qw     &       \qw     &       \qw     &       \qw&       \qw     &       \qw     &       \qw     &       \qw     &       \qw     \\
\lstick{q_{8}}  &       \qw     &       \qswap\qwx[4]   &       \qw     &       \qw     &       \qw     &       \qw     &       \qw     &       \qw     &       \qw     &       \qw     &       \qw     &       \qw&       \qw     &       \qw     &       \qw     &       \qw     &       \qw     &       \qw     &       \qw     &       \qw     &       \qw     &       \qw     &       \qw     &       \qw     &       \qw     & \qw      &       \qw     &       \qw     &       \qw     &       \qw     &       \qw     &       \qw     &       \qw     &       \qw     &       \qw     &       \qw     &       \qw     &       \qw     &       \qw&       \qw     &       \qw     &       \qw     &       \qw     &       \qw     \\
\lstick{q_{9}}  &       \qw     &       \qw     &       \gate{H}        &       \ctrl{4}        &       \ctrl{3}        &       \qw     &       \qw     &       \qw     &       \qw     &       \qw     &       \qw&       \qw     &       \gate{H}        &       \qswap  &       \qw     &       \qw     &       \qw     &       \qw     &       \qw     &       \qw     &       \qw     &       \qw     &       \qw     &       \qw&       \qw     &       \qw     &       \qw     &       \qw     &       \qw     &       \qw     &       \gate{H}        &       \qw     &       \qw     &       \qw     &       \qw     &       \qw     &       \qw&       \qw     &       \ctrl{3}        &       \ctrl{4}        &       \gate{H}        &       \qw     &       \qw     &       \qw     \\
\lstick{q_{10}} &       \qw     &       \qw     &       \gate{H}        &       \qw     &       \qw     &       \ctrl{4}        &       \ctrl{3}        &       \ctrl{2}        &       \qw     &       \qw     & \qw      &       \qw     &       \gate{H}        &       \qw     &       \qswap  &       \qw     &       \qw     &       \qw     &       \qw     &       \qw     &       \qw     &       \qw     &       \qw     & \qw      &       \qw     &       \qw     &       \qw     &       \qw     &       \qw     &       \qw     &       \gate{H}        &       \qw     &       \qw     &       \qw     &       \qw     &       \ctrl{2}  &\ctrl{3}        &       \ctrl{4}        &       \qw     &       \qw     &       \gate{H}        &       \qw     &       \qw     &       \qw     \\
\lstick{q_{11}} &       \qw     &       \qw     &       \gate{H}        &       \qw     &       \qw     &       \qw     &       \qw     &       \qw     &       \ctrl{4}        &       \ctrl{3}        &       \ctrl{2}   &       \ctrl{1}        &       \gate{H}        &       \qw     &       \qw     &       \qswap  &       \qw     &       \qw     &       \qw     &       \qw     &       \multigate{1}{\mathrm{ SWAP }}     & \qw      &       \qw     &       \qw     &       \multigate{1}{\mathrm{ SWAP }}     &       \qw     &       \qw     &       \qw     &       \qswap  &       \qw     &       \gate{H}        &       \ctrl{1}        & \ctrl{2} &       \ctrl{3}        &       \ctrl{4}        &       \qw     &       \qw     &       \qw     &       \qw     &       \qw     &       \gate{H}        &       \qw     &       \qw     &       \qw     \\
\lstick{q_{12}} &       \multigate{3}{\mathrm{ U_{D_u} }}     &       \qswap  &       \qw     &       \qw     &       \targ   &       \qw     &       \qw     &       \targ   &       \qw     &       \qw     &       \qw&       \targ   &       \qw     &       \qw     &       \qw     &       \qw     &       \qswap  &       \qw     &       \qw     &       \multigate{3}{\mathrm{ U_{D_v} }}     &       \ghost{\mathrm{ SWAP }}    &       \multigate{3}{\mathrm{ U_{{D_u}^{-1}} }}        &       \qw     &       \multigate{3}{\mathrm{  U_{D_{v^{\downarrow}}}}}     &       \ghost{\mathrm{ SWAP }}    &       \multigate{3}{\mathrm{ U_{D_{u^{\downarrow}}^{-1}}}}    &       \qw     &       \qw     &       \qw     & \qswap   &       \qw     &       \targ   &       \qw     &       \qw     &       \qw     &       \targ   &       \qw     &       \qw     &       \targ   &       \qw     &       \qw     &       \qswap  &       \multigate{3}{\mathrm{ U_{D_u} }}        &       \qw     \\
\lstick{q_{13}} &       \ghost{\mathrm{ U_{D_u} }}    &       \qw     &       \qw     &       \targ   &       \ctrl{-1}       &       \qw     &       \targ   &       \ctrl{-1}       &       \qw     &       \qw     & \targ    &       \ctrl{-1}       &       \qw     &       \qw     &       \qw     &       \qw     &       \qw     &       \qw     &       \qw     &       \ghost{\mathrm{ U_{D_v} }}    &       \qw     &       \ghost{\mathrm{ U_{{D_u}^{-1}} }}       &       \qw     &       \ghost{\mathrm{ U_{D_{v^{\downarrow}}}}}    &       \qw     &       \ghost{\mathrm{ U_{D_{u^{\downarrow}}^{-1}} }}    &       \qw     &       \qw     &       \qw     &       \qw     &       \qw     &       \ctrl{-1} &\targ   &       \qw     &       \qw     &       \ctrl{-1}       &       \targ   &       \qw     &       \ctrl{-1}       &       \targ   &       \qw     &       \qw     &       \ghost{\mathrm{ U_{D_u} }}    &       \qw\\
\lstick{q_{14}} &       \ghost{\mathrm{ U_{D_u} }}    &       \qw     &       \qw     &       \qw     &       \qw     &       \targ   &       \ctrl{-1}       &       \ctrl{-1}       &       \qw     &       \targ   & \ctrl{-1}        &       \ctrl{-1}       &       \qw     &       \qw     &       \qw     &       \qw     &       \qw     &       \qw     &       \qw     &       \ghost{\mathrm{ U_{D_v} }}    &       \qw     &       \ghost{\mathrm{  U_{{D_u}^{-1}} }}       &       \qw     &       \ghost{\mathrm{ U_{D_{v^{\downarrow}}}}}    &       \qw     &       \ghost{\mathrm{ U_{D_{u^{\downarrow}}^{-1}} }}    &       \qw     &       \qw     &       \qw     &       \qw     &       \qw     &       \ctrl{-1}  &       \ctrl{-1}       &       \targ   &       \qw     &       \ctrl{-1}       &       \ctrl{-1}       &       \targ   &       \qw     &       \qw     &       \qw     &       \qw     &       \ghost{\mathrm{ U_{D_u} }}       &       \qw     \\
\lstick{q_{15}} &       \ghost{\mathrm{ U_{D_u} }}    &       \qw     &       \qw     &       \qw     &       \qw     &       \qw     &       \qw     &       \qw     &       \targ   &       \ctrl{-1}       &       \ctrl{-1}  &       \ctrl{-1}       &       \qw     &       \qw     &       \qw     &       \qw     &       \qw     &       \qw     &       \qw     &       \ghost{\mathrm{ U_{D_v} }}    &       \qw     &       \ghost{\mathrm{ U_{{D_u}^{-1}} }}       &       \qw     &       \ghost{\mathrm{ U_{D_{v^{\downarrow}}}}}    &       \qw     &       \ghost{\mathrm{ U_{D_{u^{\downarrow}}^{-1}} }}    &       \qw     &       \qw     &       \qw     &       \qw     &       \qw     &       \ctrl{-1} &\ctrl{-1}       &       \ctrl{-1}       &       \targ   &       \qw     &       \qw     &       \qw     &       \qw     &       \qw     &       \qw     &       \qw     &       \ghost{\mathrm{ U_{D_u} }}    &       \qw\\
\end{myqcircuit}
$}
\end{center}
\caption{Example with $n=3$}
\label{fig:BE-Sexample}
\end{figure}

\begin{theorem}\label{teo:FE}
Given a one-pair matrix $S\in \mathbb{R}^{N\times N}$,  $N=2^n,$  with generators $u, v$, 
$u_i\neq 0$ for $i=0, 1, \ldots, N-1$, the circuit in Figure~\ref{fig:BE-S} implements a $\left(\frac{2N^2M_u^2M_v}{c^4 m_u},\right.$ 
$\left.2\log(N)+7,\frac{2N^2M_u}{c^2}\left( \frac{M_u}{m_u}\epsilon_v +M_v \left (M_u + \frac{2}{c m_u} \right)\epsilon_u \right)\right)-$block encoding of $S$.   
The circuit has depth $\mathcal{O}({\rm{polylog}}(N)+ {\rm{poly}}(t))$ plus 12 calls to diagonal oracles, where 

$$t=\max\left\{ \left\lceil \log_2\left(\frac{M_v \pi}{c \epsilon_v}\right) \right\rceil , \left\lceil \log_2\left( \frac{\pi \, M_u}{\epsilon_u }\right) \right\rceil +\left\lceil \log_2\left( \frac{c\,M_u}{ m_u}\right) \right\rceil\right\}+1.
$$

\end{theorem}

\begin{proof}
The proof is obtained simply combining the block encodings of $D_u$, $L$, $\Delta_z$ and $L^T$, with the help of Lemma~\ref{lemma:prod}.     
Note that the block encoding of $L^T$ is obtained running the circuit for $U_L$ in reverse order, since $U_L^\dagger=U_L^T$ is a block encoding of $L^T.$ 
From Theorem~\ref{diag_encoding} we have that $U_{D_u}$ is an $(M_u/c, 1, \epsilon_u)$-block encoding of $D_u$, and $U_L$ is a $(N, \log(N)+1, 0)-$block encoding of $L$. This allows us to construct a $(N \, M_u/c, \log(N)+2, N\epsilon_u)-$ block encoding of $D_u L$. Multiplying by $U_{\Delta_z}$, we get a $(2 N M_u M_v/(m_uc^3), \log(N)+5, 2\frac{N}{c} \left( \frac{M_u}{m_u} \epsilon_v+ M_v( M_u+\frac{1}{c m_u}) \epsilon_u\right)-$ block encoding of $D_u L \Delta_z$. 
Multiplying by $U_L^\dagger$ we obtain a block encoding of $D_u L \Delta_z L^T$.

In this step we get a  
$$
\left( \frac{2N^2 M_u}{m_u}\frac{M_v}{c^3}, \ 2\log(N) + 6, \ 2\frac{N^2}{c} \left( \frac{M_u}{m_u} \epsilon_v+ M_v( M_u+\frac{1}{c m_u}) \epsilon_u\right) \right)
$$  
$-$block encoding.  
 The final multiplication is with $U_{D_u}$ and we obtain  a block encoding with scaling coefficient  equal to
$$2N^2 \frac{M_u^2}{m_u} \frac{M_v}{c^4}= \mathcal{O}(N^2).
$$
This block encoding uses $2 \log(N)+7$ ancillary qubits overall, and yields the block encoding of  matrix $S$ with an error 
\begin{equation} \label{eq:final_error}
\frac{2N^2M_u}{c^2}\left( \frac{M_u}{m_u}\epsilon_v +M_v \left (M_u + \frac{2}{c m_u} \right)\epsilon_u \right).
\end{equation}
The depth of the circuit is obtained combining Theorems~\ref{diag_encoding}, ~\ref{teo:fast_inversion}, ~\ref{theo:Ltheorem} and~\ref{lemma:Dz}. 

\end{proof}

Assuming $\epsilon=\epsilon_u=\epsilon_v$, we get that ~\eqref{eq:final_error} becomes
$$ 2\, \epsilon \,N^2  \frac{M_u}{c^2}\left( \frac{M_uc +M_v M_u c m_u + 2 M_v}{c m_u} \right)= \mathcal{O}( N^2 \epsilon).$$
As expected, the error also grows with the ratio $M_u/m_u$ and with $M_v$.

\subsection{Some considerations about the encoding}
\label{sec:considerations}
In the definition of block encoding, we state that it is only possible to encode a matrix scaled by a factor $\alpha \ge \|A\|_2.$
Using the factorization of a one-pair  matrix $S=D_u L \Delta_z L^T D_u$ and encoding each factor we note that we cannot expect to achieve better results than those reported in  Theorem~\ref{teo:FE}. Indeed we have
$$\|S\|_2\le \|D_u\|_2^2\|L\|_2^2\|\Delta_z\|_2.$$
Since $L^{-1}={\rm diag(ones}(N, 1), 1)-{\rm diag(ones}(N, 1), -1)$, we have that 
$$
\|L\|_2 \le \frac{2N}{ \pi}.
$$

Indeed, $\|L\|_2^2 =\frac{1}{\lambda_{min}(L^{-T}L^{-1})},$ and we have that $L^{-T}L^{-1}=2 I_N-T_{N, 0, 1},$ $T_{N, 0, 1}$ being the tridiagonal matrix with ones on the upper and lower diagonal and zeros on the main diagonal except for the entry $N, N$ which is 1. The eigenvalues of this special class of matrices can be explicitly computed~\cite{Lo92} so we have that $\lambda_{min}(L^{-T}L^{-1})=2+2\cos(2\pi N/(2N+1))\ge \frac{\pi^2}{4N^2}.$

Note that $\|D_u\|_2\le M_u,$ while $\|\Delta_z\|_2\le 2 \frac{M_v}{m_u}.$ Summing up, we get the following upper bound for the norm of $S$ as
$$
\|S\|_2\le \frac{8N^2 M_u^2 M_v}{\pi^2 m_u},
$$
which means that the scaling factor resulting in Theorem~\ref{teo:FE} cannot be much improved if this factorization is used. Indeed it is asymptotically optimal  both with respect to the size of the matrix and to the conditioning of vectors $u$ and $v$. 

The scaling term could potentially be reduced by a factor of $N$ if one can block encode a one-pair  matrix without relying on the factorization employed in this paper.
However we can prove that $\|S\|_2 =\mathcal{O}(N).$ Indeed, the elements of  $ S $ must satisfy $ |S_{ij}| \leq M_u M_v $ for all $ i, j $. Consequently, the $ L^1 $ norm of $ S $ is given by:

$$\|S\|_1 \leq M_u M_v N.$$

\begin{figure}
\begin{center}

\includegraphics[width=0.70\textwidth]{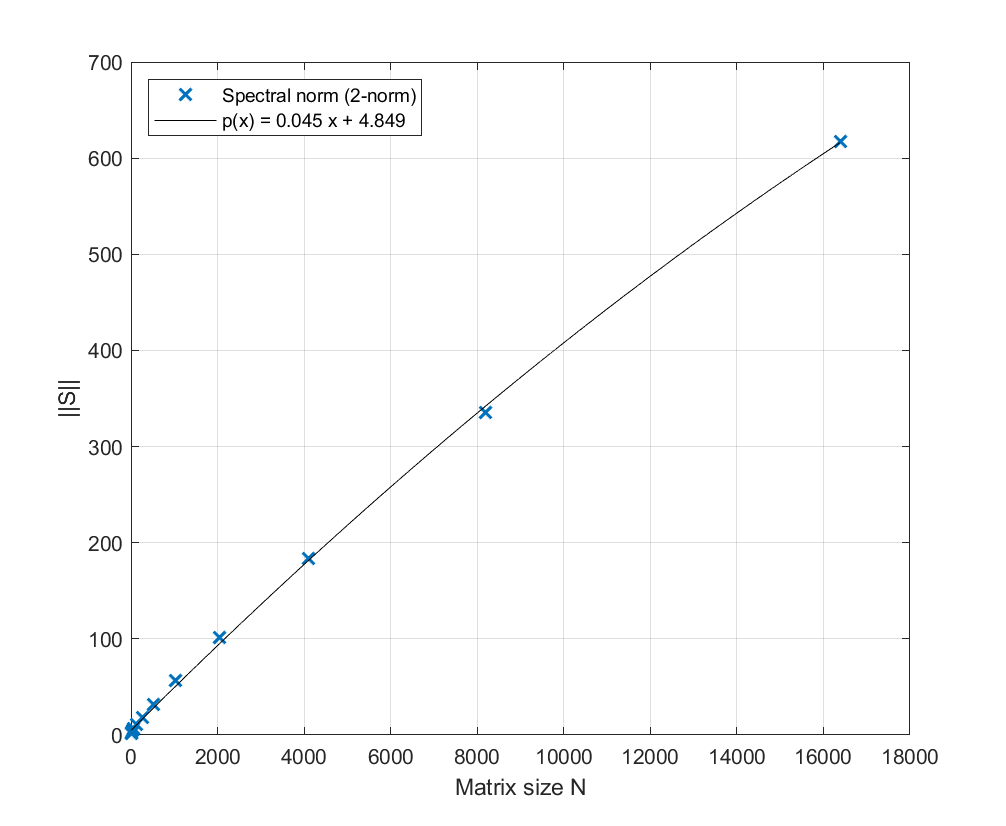}
\caption{Plot of the average $2$-norm for one-pair semiseparable matrices of different size. We clearly see the linear growth with respect to the size of the matrix. The solid line represents the polyfit curve for the 2-norm average values. The best fit polynomial is $p(x)=0.0450 x+  4.8493$.  }
\label{fig:plotnorm}
\end{center}
\end{figure}

Since $ S $ is symmetric, its $2$-norm coincides with its spectral radius. It is known that the spectral radius of a matrix is less than or equal to any induced norm. Therefore, we have:

$$\|S\|_2 \leq M_u M_v N.$$

On the other hand, the matrix $ \text{ones}(N) $ is semiseparable and has a $2$-norm exactly equal to $ N $ and $M_u=M_v=1$. Thus, the inequality derived above is sharp, and we can conclude that the $2$-norm of a generic one-pair  matrix with generators $ u $ and $ v $ grows as $ \mathcal{O}(M_u M_v N) $. 
Figure~\ref{fig:plotnorm} shows a plot of the $2$-norm of random one-pair semiseparable matrices of size ranging from  2 to $2^{14}$. The generators have been chosen randomly with distribution $\mathcal{N}(0, 1)$ and then normalized so that $\|u\|_{\infty}=\|v\|_{\infty}=1.$
Each point is the result of an average over 500 runs of one-pair semiseparable random matrices of the same size.

The scaling factor $\alpha= \mathcal{O}(N^2)$ introduces a multiplicative overhead in algorithms whose cost scales linearly with $\alpha$, such as Hamiltonian simulation and quantum linear solvers~\cite[Corollary~60]{gilyen2019quantum}. However, this overhead may be acceptable when circuit depth is the dominant bottleneck, or when only relative spectral information is required: for instance, when computing eigenvalues of the encoded matrix, the scaling factor can be removed by classical post-processing at no additional quantum cost.

\begin{figure}
\begin{center}
\includegraphics[width=0.65\textwidth]{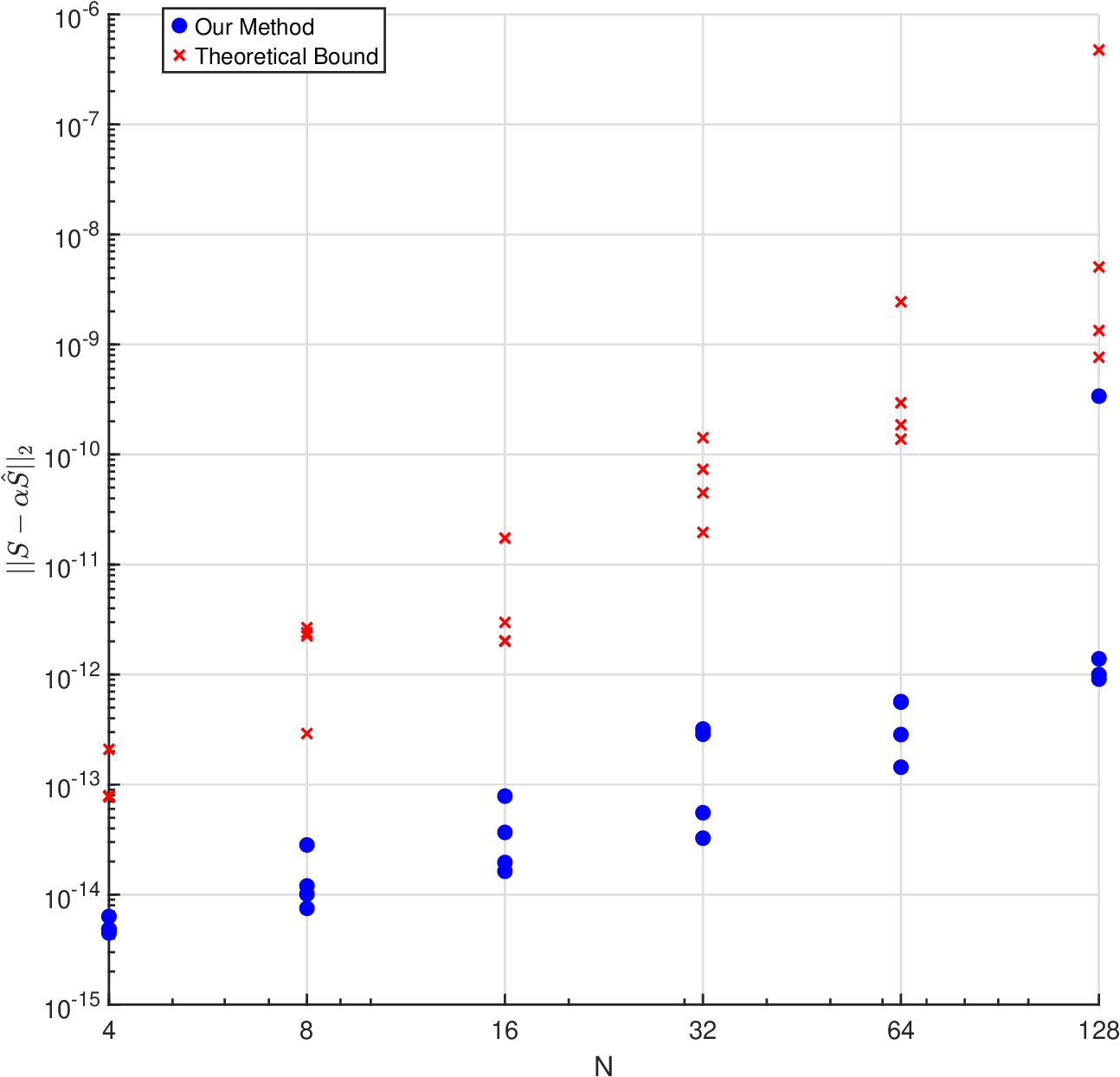}
\caption{Comparison between the theoretical and actual error for randomly generated matrices of size ranging from $4 \times 4$ to $128 \times 128$. The scaling factor $\alpha$ is as reported in Theorem~\ref{teo:FE}.}

\label{Fig:comp_theo}
\end{center}
\end{figure}
\section{Numerical experiments}\label{Sec:NE}
In this section, we report experiments comparing the performance of our encoding scheme against the theoretical error bound provided in Theorem~\ref{teo:FE}, as well as with the encoding obtained using FABLE~\cite{fable}, a general-purpose block-encoding algorithm.  All tests were performed using QCLAB~\cite{qclab}, an open-source MATLAB package for the simulation of quantum circuits.

Before reporting the results of our comparison, we need to make some observations: 
\begin{itemize}
\item To the best of our knowledge, there are no other works in the literature, except the present one, that specifically address the block-encoding of one-pair  matrices. For this reason we selected as reference benchmark a method for block-encoding general matrices, such as FABLE. 
\item FABLE directly encodes the matrix entries using 
controlled rotation gates, 
and  the full cost of loading all $N^2$ entries is accounted for in 
the circuit complexity, resulting in a circuit of depth $\mathcal{O}(N^2)$ 
without compression. Our method, on the contrary accesses the matrix entries using the diagonal oracles, and the circuit 
complexity beyond the oracle calls is $\mathcal{O}({\rm{polylog}}(N))$. 
A direct comparison with FABLE in terms of total gate count is only meaningful when the implementation cost of the oracle $O_D$ is taken into account. For general, unstructured input data (the generators $u$ and $v$), implementing $O_D$ requires a circuit of depth $\Omega(N)$,  which brings the total complexity to $\Omega(N)$, only a factor of $\mathcal{O}(N)$ better than FABLE but far from an exponential improvement. An exponential improvement in circuit complexity over FABLE is achievable only when we admit that the data are stored in QRAM or the matrix possesses sufficient structure to allow $O_D$ to be realized in $\mathcal{O}({\rm polylog}(N))$ gates. We discuss this situation  in Remark~\ref{parag:oracles}. In the absence of such conditions, 
the two methods address different data-access models 
and are not directly comparable in terms of circuit complexity.
\item Note however, that the error of the two compared methods originates from different sources. In particular the error of  our method derives from representation error of the generators: we assume to access the generators entries through the oracle with machine precision $\epsilon$, i.e., $\epsilon_v=\epsilon_u\approx 10^{-16}$. Then the errors combine additively and multiplicatively  from the operations applied to the generators to achieve the final error in block-encoding reported in Theorem~\ref{teo:FE}. On the other hand, in FABLE, the entries of the matrix are initially represented up to machine precision, and  the transformation to and from the Walsh-Hadamard space, as well as  the threshold value used as a cut-off parameter, introduce a compression error.

\end{itemize}

\begin{figure}
\begin{center}
\includegraphics[width=0.65\textwidth]{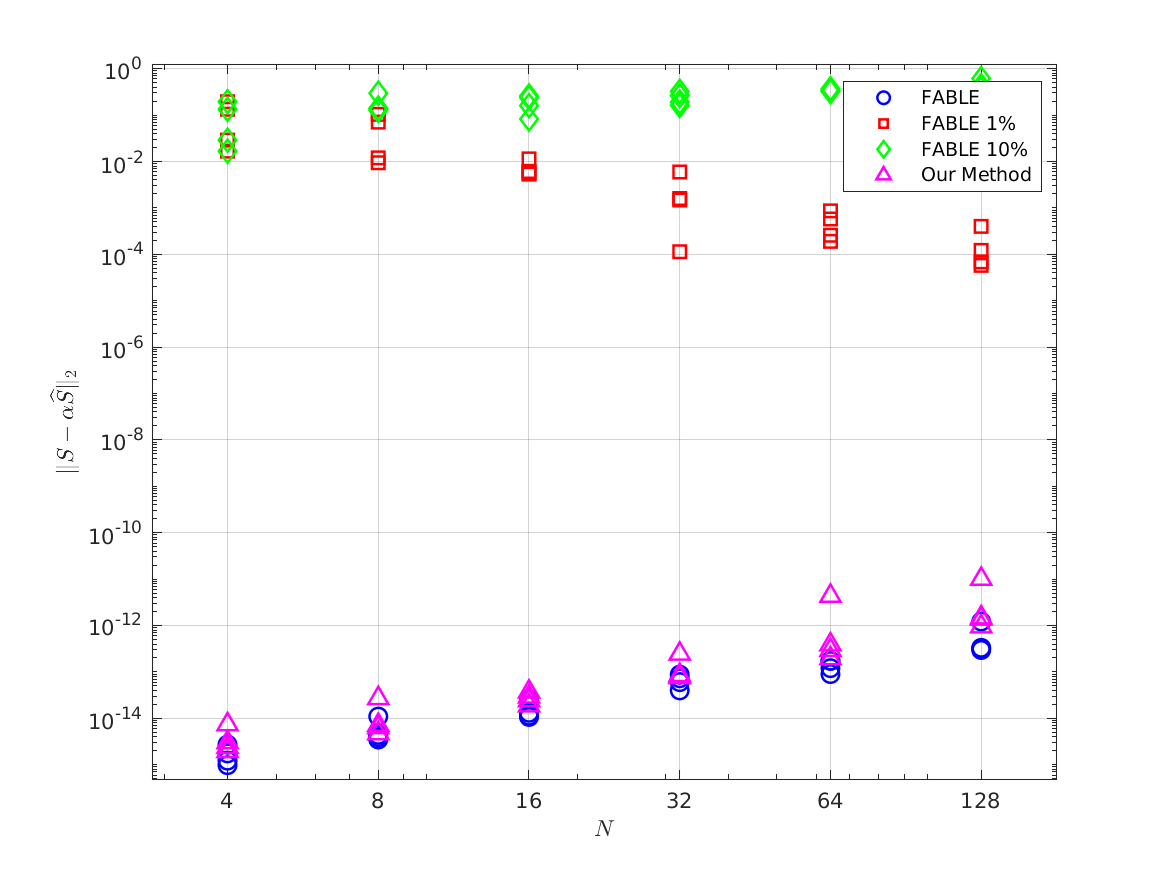}
\caption{Comparison of the encoding error obtained by our method and FABLE without compression (FABLE) and with different compression rates for randomly generated matrices of size ranging from $4\times 4$ to $128\times 128$.  The results reported as FABLE $1\%$ correspond to the removal of just one rotation, while FABLE $10\%$ removes $10\%$ of rotations. }
\label{Fig:comp_Fable}
\end{center}
\end{figure}

In Figure~\ref{Fig:comp_theo}, we compare the error in the block encoding of an $N \times N$ one-pair matrix against the theoretical bound derived for its block encoding and reported in Theorem~\ref{teo:FE}. For $N=4,8,16,32,64,128$ we generate three random one-pair matrices
whose generators are sampled from a normal distribution with mean 0 and variance 1. For each case, let $S$ be the original one-pair matrix, and denote by $\hat{S}$ the principal submatrix of the unitary obtained from the circuit rescaled according to the normalized parameter $\alpha$; then the computed error is defined as
$\|S-\alpha \hat{S}\|_2$.
 As shown in Figure~\ref{Fig:comp_theo}, the actual error of our encodings remains consistently below the theoretical bounds by several orders of magnitude.

In Figure~\ref{Fig:comp_Fable}, we compare the performance of our encoding scheme with FABLE under three different approximation settings. The first comparison considers the exact version of FABLE, where no compression is applied and the number of rotations required to encode the matrix scales as $\mathcal{O}(N^2)$. In this configuration, our method achieves a comparable level of accuracy while maintaining a complexity of $\mathcal{O}({\rm polylog}(N))$ plus two calls to the oracle $O_D$. We point out that our algorithm requires a slightly higher number of ancillary qubits, precisely $2\log(N)+7$ versus the $\log(N)+1$ required by FABLE, so the classical simulation of our circuit requires a longer computation time, but on a real quantum hardware with enough qubits the computation of our circuit 
can be exponentially faster in cases where the oracle $O_D$ is implementable with polylogarithmic depth.
Moreover, some of our ancillary qubits can be freed and can therefore be overlooked in the overall storage count.

Note that using a general-purpose encoding such as FABLE can be detrimental for structured matrices. The Walsh–Hadamard transform (WHT), used to increase data sparsity, together with the entrywise encoding and rounding, can break the low-rank generator structure of the original matrix. In Figure~\ref{Fig:error_inv}(a) we plot the spectral norm distance between the inverse of the original one-pair semiseparable matrix S and the scaled inverse of its block-encoded approximation, i.e. $\|S^{-1} - \hat S^{-1}/\alpha \|_2$. 
Even without compression, FABLE does not preserve the semiseparable structure: encoding the entries as ${\rm round}(u_iv_j)$ destroys the generator representation, and the error in the inverse is larger than the one computed by or method. Indeed, our method better preserves the structure because it is enforced by the factorization we use. 

The accuracy of the inverse is however sensitive to the conditioning of $S$. Using the standard matrix perturbation bound, if $\|S -\alpha \hat S\|_2 \le \delta$, then
$$\|S^{-1} - \frac{1}{\alpha}\hat S^{-1}\|_2\le \kappa(S)  \frac{\|S^{-1} \|}{\|S\|}\, \delta= \|S^{-1}\|^2 \,\delta.$$

This shows that the error in the inverse scales with the square of $\|S^{-1}\|_2$, and hence grows rapidly for ill-conditioned matrices. This explains the larger errors observed for our method at $N=128$, where the randomly generated test matrices have condition numbers ranging from $\mathcal{O}(10^3)$ to $\mathcal{O}(10^6)$.
 In Figure~\ref{Fig:error_inv}(b) we report the quantity $\frac{1}{\alpha}\|\hat S^{-1}-{\rm tridiag}(\hat{S}^{-1})\|_F$, where ${\rm tridiag}(A)$ denotes the matrix with entries $A_{i,j}$ for $|i-j|\le 1$
and zeros elsewhere. Since the inverse of a one-pair semiseparable matrix is tridiagonal, this metric measures how closely $\hat{S}$ (via its inverse) retains the one-pair semiseparable structure.  
\begin{figure}[htbp]
\centering
\begin{subfigure}[t]{0.49\textwidth}
    \centering
    \includegraphics[width=\textwidth]{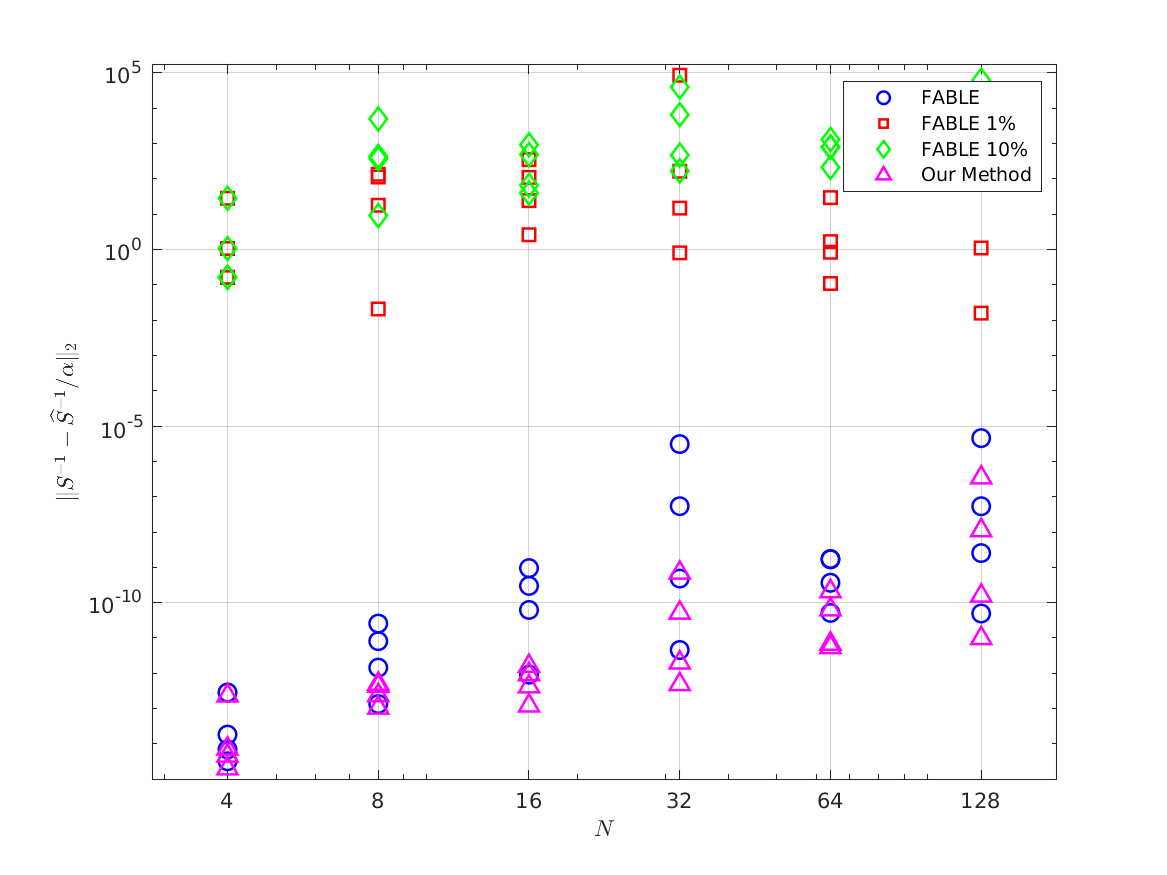}
    \caption{Error in the inverse of the scaled block encoded matrix, quantified as $\|S^{-1} - \hat{S}^{-1}/\alpha\|_2$.}
    \label{fig:plot3a}
\end{subfigure}
\hfill
\begin{subfigure}[t]{0.49\textwidth}
    \centering
    \includegraphics[width=\textwidth]{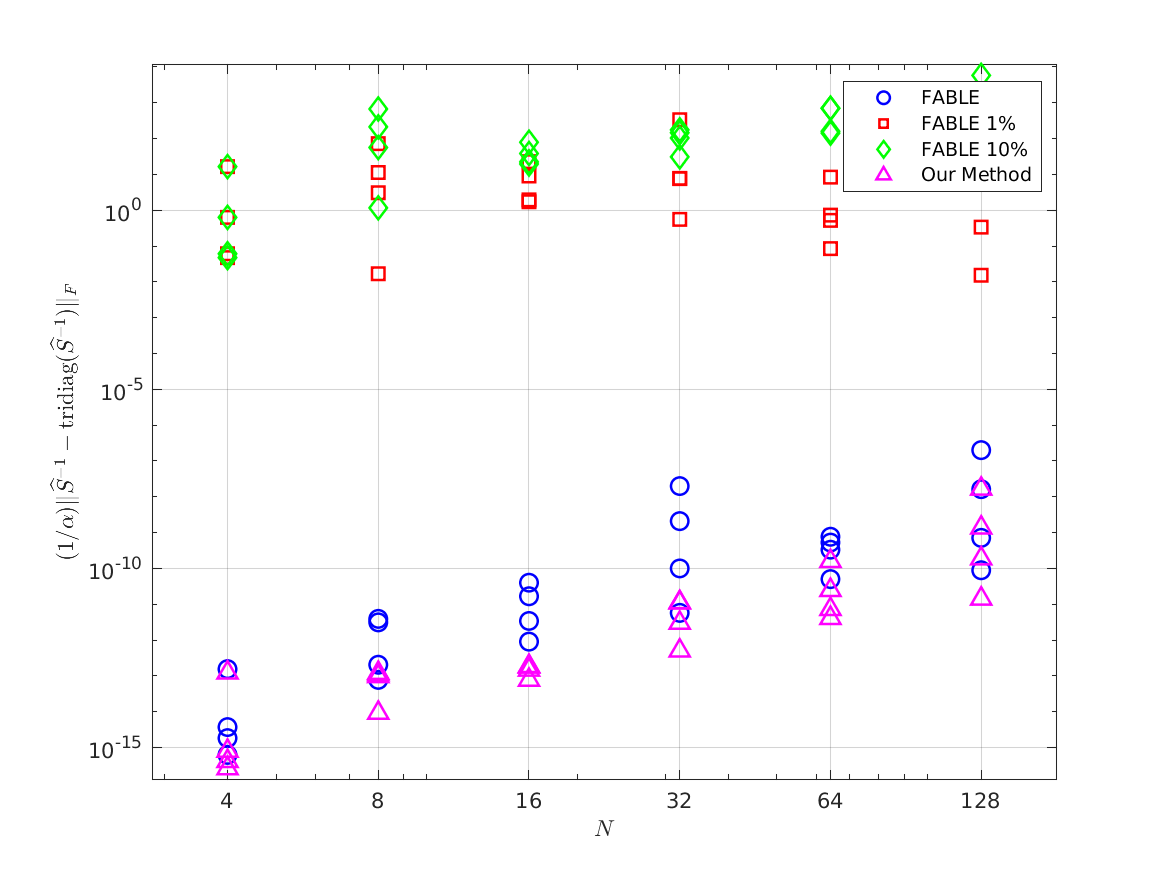}
    \caption{Error indicating the deviation of the computed inverse from the tridiagonal structure, expressed as $\frac{1}{\alpha}\|\hat{S}^{-1} - \text{tridiag}(\hat{S}^{-1})\|_F$. }
    \label{fig:plot3b}
\end{subfigure}
\caption{Comparison of the inverses of the block encoded matrices $\alpha\hat{S}$ with the ground truth $S$ and computation of the norm of the off-tridiagonal elements. Knowing that the inverse of a one-pair semiseparable is a tridiagonal matrix, we expect the norm of the off-tridiagonal elements to be null.}
\label{Fig:error_inv}
\end{figure}

Overall, although the exact FABLE encoding achieves a slightly lower representation error of the one-pair semiseparable matrix, our method better preserves its structural characteristics. In particular, despite exhibiting a marginally higher representation error, our encoding yields smaller errors in both the norm of the inverse and the norm of the off-tridiagonal elements. We note that, using the FABLE encoding scheme, even without compression, the one-pair semiseparable structure is lost. Specifically, FABLE is applied directly to the $N^2$ entries of the matrix, which
it treats as independent parameters. The WHT is a global, dense transformation: each output coefficient is a linear combination of all $N^2$ input values, with coefficients $\pm 1$. During this transformation some error is introduced along all the values, while in the proposed method the error of the result is always dependent on the original generators. This means that while the inversion of the final encoded matrix of our proposed method is affected by error deriving only from the generators, the inversion of the FABLE results is affected by errors over the whole matrix representation. This disrupts the low-rank properties of the matrix and affects the off-tridiagonal entries of the inverse.

 \section{Conclusions}   
Block-encoding is a fundamental task in quantum numerical linear algebra, as it provides an embedding of a given matrix into a unitary suitable for quantum computation. 

In this work we have explored the role of matrix rank structure -- namely, one-pair semiseparable structure -- in the design of an efficient block-encoding procedure. Relying on a specific factorization that characterizes one-pair semiseparable matrices, we propose a block-encoding method that requires polylogarithmic time, assuming that one-pair semiseparable generators can also be encoded in polylogarithmic time. 
 Experimental comparison with state-of-the-art general-purpose encoding software FABLE shows that the accuracy of our approach is  comparable to FABLE without compression and remarkably better than FABLE when carried out with several compression levels.  

 In classical numerical linear algebra, it is well-known that sparse and data-sparse structure allows for the design of linear complexity algorithms for computing matrix-vector products or solving linear systems. In a quantum setting, if an efficient input model is available, one may expect to achieve polylogarithmic complexity in block encoding of rank-structured matrices, which in turn may be exploited towards efficient quantum linear algebra algorithms. An important open problem concerns the subnormalization factor of our encoding, which scales as $\mathcal{O}(N^2)$
 due to the factorization approach employed. As shown in Section~\ref{sec:considerations}, the spectral norm of a generic one-pair matrix grows as $\mathcal{O}(N)$, suggesting that the current subnormalization factor is suboptimal by a factor of $\mathcal{O}(N)$. Closing this gap would require a block encoding strategy that does not rely on the triangular-diagonal factorization of Proposition~\ref{prop:decomoposition}, and represents a natural direction for future work. A tighter subnormalization factor would directly improve the performance of quantum algorithms applied to semiseparable matrices, such as Hamiltonian simulation and quantum linear solvers, since the  cost scales linearly with $\alpha$. 

The present work can be seen as a first step towards efficient block-encoding strategies for more complex matrix structures, such as hierarchical semiseparable~\cite{Xia10}.

\section*{Acknowledgments}

This study was carried out within the National Centre on HPC, Big Data and Quantum Computing - SPOKE 10 (Quantum Computing) and received funding from the European Union Next-GenerationEU - National Recovery and Resilience Plan (NRRP) – MISSION 4 COMPONENT 2, INVESTMENT N. 1.4 – CUP N. I53C22000690001. 

Partial support was also provided by the INdAM - GNCS project ``Strutture di matrici e di funzioni per la sintesi di circuiti quantistici efficienti'', CUP E53C23001670001.

P. Boito acknowledges the MIUR Excellence Department Project awarded to the Department of Mathematics, University of Pisa, CUP I57G22000700001.

P. Boito and G. Del Corso are also partially supported by European Union - NextGenerationEU under the National Recovery and Resilience Plan (PNRR) - Mission 4 Education and research - Component 2 From research to business - Investment 1.1 Notice Prin 2022 - DD N. 104 2/2/2022, titled Low-rank Structures and Numerical Methods in Matrix and Tensor Computations and their Application, proposal code 20227PCCKZ – CUP I53D23002280006.
J53D23003620006
P. Boito and G. Del Corso are also partially supported by  the U.S. Department of Energy, Office of Science, National Quantum Information Science Research Centers, Superconducting Quantum Materials and Systems Center (SQMS) under contract number DE-AC02-07CH11359.

The work of M. Porcelli is partially supported by INdAM - GNCS project CUP E53C24001950001.

\subsection*{Declarations}
\subsubsection*{Author contributions} GDC focused on the development of the quantum algorithm with contributions from GA, PB and MP.   PB and MP handled  the writing of the paper with contribution from GA and GDC.  GA was responsible for the development of the code and for the numerical testing.  All four authors carefully reviewed the manuscript and discussed the theoretical aspects as well as the approach to solving the problem.
\subsubsection*{Ethical Approval}
It complies with the Ethical Rules applicable for this journal.
\subsubsection*{Data Availability}
No datasets were generated or analyzed during the current study.
\subsubsection*{Code Availability} 
The code of the implemented algorithm is available at\\ \url{https://github.com/Giacomo-Antonioli/QBES}.
\subsubsection*{Competing interests}
The authors declare no competing interests.
\bibliographystyle{plain}
\bibliography{biblio}

@article{wan2021block,
  title={Block-encoding-based quantum algorithm for linear systems with displacement structures},
  author={Wan, Lin-Chun and Yu, Chao-Hua and Pan, Shi-Jie and Qin, Su-Juan and Gao, Fei and Wen, Qiao-Yan},
  journal={Physical Review A},
  volume={104},
  number={6},
  pages={062414},
  year={2021},
  publisher={APS}
}

@article{camps2024explicit,
  title={Explicit quantum circuits for block encodings of certain sparse matrices},
  author={Camps, Daan and Lin, Lin and Van Beeumen, Roel and Yang, Chao},
  journal={SIAM Journal on Matrix Analysis and Applications},
  volume={45},
  number={1},
  pages={801--827},
  year={2024},
  publisher={SIAM}
}

@inproceedings{fable,
  title={{FABLE}: Fast approximate quantum circuits for block-encodings},
  author={Camps, Daan and Van Beeumen, Roel},
  booktitle={2022 IEEE International Conference on Quantum Computing and Engineering (QCE)},
  pages={104--113},
  year={2022},
  organization={IEEE}
}

@book{KLM,
  title        = {An Introduction to Quantum Computing},
  author       = {Kaye, Phillip and Laflamme, Raymond and Mosca, Michele},
  publisher    = {Oxford University Press},
  year         = {2007},
  address      = {Oxford},
  isbn         = {0198570007},
  pages        = {288}
}

@article{takahira2021quantum,
  title={Quantum algorithms based on the block-encoding framework for matrix functions by contour integrals},
  author={Takahira, Souichi and Ohashi, Asuka and Sogabe, Tomohiro and Usuda, Tsuyoshi Sasaki},
  journal={arXiv preprint arXiv:2106.08076},
  year={2021}
}

@article{vandebril2005note,
  title={A note on the representation and definition of semiseparable matrices},
  author={Raf Vandebril  and Marc Van Barel and Nicola Mastronardi},
  journal={Numerical Linear Algebra with Applications},
  volume={12},
  number={8},
  pages={839--858},
  year={2005},
  publisher={Wiley Online Library}
}

@article{li2023efficient,
  title={On efficient quantum block encoding of pseudo-differential operators},
  author={Haoya Li  and  Hongkang Ni and  Lexing Ying},
  journal={Quantum},
  volume={7},
  pages={1031},
  year={2023},
  publisher={Verein zur F{\"o}rderung des Open Access Publizierens in den Quantenwissenschaften}
}

@misc{fasino,
  title={A brief introduction to quasiseparable matrices},
  author={Dario Fasino},
    year =2011,
  note={https://www.mat.uniroma2.it/$\sim$tvmsscho/Rome-Moscow\_School/2011/files/quasisep-fasino-notes.pdf}
}

@article{lin2022lecture,
  title={Lecture notes on quantum algorithms for scientific computation},
  author={Lin Lin},
  journal={arXiv preprint arXiv:2201.08309},
  year={2022}
}

@inproceedings{gilyen2019quantum,
  title={Quantum singular value transformation and beyond: exponential improvements for quantum matrix arithmetics},
  author={Andr{\'a}s Gily{\'e}n and Yuan Su  and  Guang Hao Low and  Nathan Wiebe},
  booktitle={Proceedings of the 51st Annual ACM SIGACT Symposium on Theory of Computing},
  pages={193--204},
  year={2019}
}

@article{harrow2009quantum,
  title={Quantum algorithm for linear systems of equations},
  author={ Aram W. Harrow and  Avinatan Hassidim and Seth Lloyd},
  journal={Physical review letters},
  volume={103},
  number={15},
  pages={150502},
  year={2009},
  publisher={APS}
}

@article{wossnig2018quantum,
  title={Quantum linear system algorithm for dense matrices},
  author={Leonard Wossnig and Zhikuan Zhao and Anupam Prakash},
  journal={Physical review letters},
  volume={120},
  number={5},
  pages={050502},
  year={2018},
  publisher={APS}
}

@book{gantmacher2002oscillation,
  title={Oscillation matrices and kernels and small vibrations of mechanical systems: revised edition},
  author={Gantmacher, Felix R and Krein, Mark G},
  year={2002},
  publisher={American Mathematical Society Rhode Island}
}

@book{vandebril2008matrix,
  title={Matrix computations and semiseparable matrices: linear systems},
  author={Vandebril, Raf and Van Barel, Marc and Mastronardi, Nicola},
  volume={1},
  year={2008},
  publisher={JHU Press}
}

@book{eidelman2014separable,
  title={Separable type representations of matrices and fast algorithms},
  author={Eidelman, Yuli and Gohberg, Israel and Haimovici, Iulian},
  year={2014},
  publisher={Springer}
}

@article{massei2020hm,
  title={{HM}-toolbox: Matlab software for {HODLR} and {HSS} matrices},
  author={Massei, Stefano and Robol, Leonardo and Kressner, Daniel},
  journal={SIAM Journal on Scientific Computing},
  volume={42},
  number={2},
  pages={C43--C68},
  year={2020},
  publisher={SIAM}
}

@book{bebendorf2008hierarchical,
  title={Hierarchical matrices},
booktitle={Lecture Notes in Computational Science and Engineering},
  author={Bebendorf, Mario},
  year={2008},
  publisher={Springer}
}

@article{fasino2002structural,
  title={Structural and computational properties of possibly singular semiseparable matrices},
  author={Fasino, Dario and Gemignani, Luca},
  journal={Linear Algebra and its Applications},
  volume={340},
  number={1-3},
  pages={183--198},
  year={2002},
  publisher={Elsevier}
}

@article{Tong2021,
  title = {Fast inversion, preconditioned quantum linear system solvers, fast Green's-function computation, and fast evaluation of matrix functions},
  author = {Tong, Yu and An, Dong and Wiebe, Nathan and Lin, Lin},
  journal = {Phys. Rev. A},
  volume = {104},
  issue = {3},
  pages = {032422},
  numpages = {33},
  year = {2021},
  month = {Sep},
  publisher = {American Physical Society},
  doi = {10.1103/PhysRevA.104.032422},
  url = {https://link.aps.org/doi/10.1103/PhysRevA.104.032422}
}

@book{AB09,
author = {Arora, Sanjeev and Barak, Boaz},
title = {Computational Complexity: A Modern Approach},
year = {2009},
isbn = {0521424267},
publisher = {Cambridge University Press},
address = {USA},
edition = {1st},
}

@article {Lo92,
    AUTHOR = {Losonczi, L.},
     TITLE = {Eigenvalues and eigenvectors of some tridiagonal matrices},
   JOURNAL = {Acta Math. Hungar.},
  FJOURNAL = {Acta Mathematica Hungarica},
    VOLUME = {60},
      YEAR = {1992},
    NUMBER = {3-4},
     PAGES = {309--322},
      ISSN = {0236-5294,1588-2632},
   MRCLASS = {15A57 (15A18)},
  MRNUMBER = {1177259},
MRREVIEWER = {Marc\ Artzrouni},
       DOI = {10.1007/BF00051649},
       URL = {https://doi.org/10.1007/BF00051649},
}

@article{Ch12,
author ={Andrew M. Childs and Nathan Wiebe},
title = {Hamiltonian Simulation Using Linear Combinations of Unitary Operations},
volume={12},
   ISSN={1533-7146},
   url={http://dx.doi.org/10.26421/QIC12.11-12},
   DOI={10.26421/qic12.11-12},
   number={11 \& 12},
   journal={Quantum Information and Computation},
   publisher={Rinton Press},
   year={2012},
   month=nov }

@ARTICLE{ChanGu06,
	author = {Chandrasekaran, S. and Gu, M. and Pals, T.},
	title = {A fast {ULV} decomposition solver for hierarchically semiseparable representations},
	year = {2006},
	journal = {SIAM Journal on Matrix Analysis and Applications},
	volume = {28},
	number = {3},
	pages = {603 – 622},
	
}

@misc{qclab,
      title={QCLAB: A Matlab Toolbox for Quantum Computing}, 
      author={Sophia Keip and Daan Camps and Roel Van Beeumen},
      year={2025},
      eprint={2503.03016},
      archivePrefix={arXiv},
      primaryClass={quant-ph},
      url={https://arxiv.org/abs/2503.03016}, 
}

@article{Xia10,
author = {Xia, Jianlin and Chandrasekaran, Shivkumar and Gu, Ming and Li, Xiaoye S.},
title = {Fast algorithms for hierarchically semiseparable matrices},
journal = {Numerical Linear Algebra with Applications},
volume = {17},
number = {6},
pages = {953-976},

year = {2010}
}

@article{brassard2000quantum,
  title={Quantum amplitude amplification and estimation},
  author={Brassard, Gilles and Hoyer, Peter and Mosca, Michele and Tapp, Alain},
  journal={arXiv preprint quant-ph/0005055},
  year={2000}
}

@article{giovannetti2008quantum,
  title={Quantum random access memory},
  author={Giovannetti, Vittorio and Lloyd, Seth and Maccone, Lorenzo},
  journal={Physical review letters},
  volume={100},
  number={16},
  pages={160501},
  year={2008},
  publisher={APS}
}

@inproceedings{kerenidis2017quantum,
  title={Quantum Recommendation Systems},
  author={Kerenidis, Iordanis and Prakash, Anupam},
  booktitle={8th Innovations in Theoretical Computer Science Conference (ITCS 2017)},
  pages={49--1},
  year={2017},
  organization={Schloss Dagstuhl--Leibniz-Zentrum f{\"u}r Informatik}
}

@book{nielsen2010quantum,
  title={Quantum computation and quantum information},
  author={Nielsen, Michael A and Chuang, Isaac L},
  year={2010},
  publisher={Cambridge university press}
}

@book{Muller2018,
  author    = {Muller, Jean-Michel and Brunie, Nicolas and de Dinechin, Florent and Jeannerod, Claude-Pierre and Joldes, Mioara and Lef{\`e}vre, Vincent and Melquiond, Guillaume and Revol, Nathalie and Torres, Serge},
  title     = {Handbook of Floating-Point Arithmetic},
  edition   = {2nd},
  publisher = {Birkh{\"a}user},
  year      = {2018},
  address   = {Cham, Switzerland},
  isbn      = {978-3-319-76525-9},
  doi       = {10.1007/978-3-319-76526-6}
}

@article{Brent1976,
  author  = {Brent, Richard P.},
  title   = {Fast Multiple-Precision Evaluation of Elementary Functions},
  journal = {Journal of the ACM},
  volume  = {23},
  number  = {2},
  pages   = {242--251},
  year    = {1976},
  month   = apr,
  doi     = {10.1145/321941.321944},
  publisher = {ACM New York, NY, USA}
}

@book{kitaev2002,
  author    = {Kitaev, Alexei Yu. and Shen, Alexander and Vyalyi, Mikhail N.},
  title     = {Classical and Quantum Computation},
  series    = {Graduate Studies in Mathematics},
  volume    = {47},
  publisher = {American Mathematical Society},
  address   = {Providence, RI},
  year      = {2002},
  pages     = {257},
  isbn      = {9780821832295}
}

\end{document}